\documentclass[journal]{IEEEtran}
\usepackage{amsmath,amsfonts}
\usepackage{algorithm,enumitem}

\usepackage[noend]{algorithmic} 
\usepackage{array}
\usepackage[caption=false,font=normalsize,labelfont=sf,textfont=sf]{subfig}
\usepackage{textcomp}
\usepackage{stfloats}
\usepackage{url}
\usepackage{verbatim}
\usepackage{graphicx,xcolor}
\usepackage{cite}
\usepackage{tikz}
\usepackage{mathtools}
\usepackage[braket]{qcircuit}
\usepackage[colorlinks,pdfstartview=FitH,citecolor=blue,linkcolor=blue,urlcolor=blue]{hyperref}
\usepackage{url}
\usepackage{colonequals}
\usepackage{braket}
\usepackage{amsthm}
\usepackage{tikzpeople}
\floatname{algorithm}{}

\allowdisplaybreaks
\newcommand{\nnote}[1]{\textcolor{brown}{[*** Nic: #1 ***]}}

\newcommand{\enote}[1]{\textcolor{red}{[*** Emina: #1 ***]}}

\newcommand{\mnote}[1]{\textcolor{blue}{[*** Michael: #1 ***]}}

\newcommand{\rev}[1]{\textcolor{blue}{#1}}
\renewcommand{\rev}[1]{#1}

\renewcommand{\nnote}[1]{}

\renewcommand{\enote}[1]{}

\renewcommand{\mnote}[1]{}

\newcommand{\F}{\mathbb{F}}
\newcommand{\Z}{\mathbb{Z}}
\newcommand{\bbF}{\mathbb{F}}
\newcommand{\C}{\mathbb{C}}

\newcommand{\N}{\mathbb{N}}

\newcommand{\ul}[1]{\underline{#1}}
\newcommand{\pr}[1] {{\left( #1 \right)}}
\renewcommand{\sb}[1] {{\left[ #1 \right]}}

\DeclareMathOperator{\supp}{Supp}
\DeclareMathOperator{\CS}{CS}

\DeclareMathOperator{\Tr}{Tr}
\DeclareMathOperator{\CNOT}{CNOT}

\usepackage{thmtools}
\declaretheorem{theorem, definition, proposition, lemma, corollary, example, question}
\declaretheorem[numbered=no]{remark, notation}
\usepackage[capitalise]{cleveref}

\begin{document}

\title{Optimal Strategies for Winning 
Certain Coset-Guessing Quantum Games}

\author{Michael Schleppy,
~\IEEEmembership{Student Member,~IEEE,}
Emina Soljanin,
~\IEEEmembership{Fellow,~IEEE,}
and Nicolas Swanson,~\IEEEmembership{Student Member,~IEEE}
\thanks{M.~Schleppy and E.~Soljanin are with the Department of Electrical and Computer Engineering, Rutgers, the State University of New Jersey, Piscataway, NJ 08854, USA, e-mail: \{michael.schleppy, emina.soljanin\}@rutgers.edu). Nicolas Swanson is at the Combinatorics and Optimization department, the University of Waterloo, Waterloo, ON, Canada, e-mail: nswanson@uwaterloo.ca.
This  research  is in part based  upon  work  supported  by  the  National  Science  Foundation  under  Grant  \# FET-2007203}
}


\maketitle

\begin{abstract}
In a recently introduced coset guessing game, Alice plays against Bob and Charlie, aiming to meet a joint winning condition. Bob and Charlie can only communicate before the game starts to devise a joint strategy.
The game we consider begins with Alice preparing a $2m$-qubit quantum state based on a random selection of three parameters. She sends the first $m$ qubits to Bob and the rest to Charlie, and then reveals to them her choice for one of the parameters. Bob is supposed to guess one of the hidden parameters, Charlie the other, and they win if both guesses are correct. From previous work, we know that the probability of Bob's and Charlie's guesses being simultaneously correct goes to zero exponentially as $m$ increases.
We derive a tight upper bound on this probability and show how Bob and Charlie can achieve it. While developing an optimal strategy, we devised an encoding circuit using only $\CNOT$ and Hadamard gates, \rev{which builds CSS codes from EPR pairs using only local operations.} We found that the role of quantum information that Alice communicates to Bob and Charlie is to make their responses correlated rather than improve their individual (marginal) correct guessing rates. 

\end{abstract}

\section{Introduction}


Quantum information science (QIS) promises to revolutionize the world of computing and information processing by exploiting quantum mechanics' unique properties. Traditional interest in QIS has revolved around computational speedups using quantum algorithms (i.e., Shor's factoring algorithm \cite{shor1994algorithms}) and information transmission (i.e., quantum teleportation \cite{bennett1993teleporting}). These methods utilize phenomena such as superposition and entanglement, which are present in the quantum setting but not in the classical one. 

The postulates of quantum mechanics also introduce restrictions, such as how quantum information is processed. A famous consequence is the no-cloning theorem (stemming from the unitary processing restriction), which refutes the existence of a universal cloning gate. 
Various information-securing schemes developed over the years exploit the no-cloning theorem to control how classical information encoded in quantum states can be accessed and distributed. 

One of the early discovered implications of the no-cloning theorem is the principle of no passive eavesdropping in quantum communications, which underpins the security of protocols such as quantum key distribution \cite{bennett2014quantum} and blind quantum computing \cite{Broadbent_2009}.
Another renowned, early security application of the no-cloning theorem is Wiesner's quantum money scheme \cite{wiesner1983conjugate}, revisited and extended by Aaronson and Christiano \cite{Qmoney:AaronsonC13}.
There are several related and more recent cryptographic schemes, e.g., \cite{coladangelo2021hidden,ananth2023cloning} and references therein. 

A cryptographically useful notion, first considered in \cite{coladangelo2021hidden,QK:VidickZ21} is that of a {\it coset state} defined as follows: 
Given a $k$-dimensional subspace $W\in \bbF_2^n$ and vectors $x,z \in \bbF_2^n$, we define the \textit{coset state} $\ket{W_{x,z}}\in (\C^2)^{\otimes n}$ by
\begin{equation}
    \ket{W_{x,z}} = \frac{1}{\sqrt{2^k}} \sum_{u \in W} (-1)^{z \cdot u} \ket{x + u}.
\end{equation} 
The state $\ket{W_{x,z}}$  depends on $W$, the coset $W + x$ and the coset $W^\perp + z$; hence the name coset state. \rev{The reader familiar with CSS codes will recognize that coset states correspond exactly to a CSS code's code or error states.}
Suppose a transmitter, Alice, prepares $\ket{W_{x,z}}$ for $k=m$ and $n=2m$ and sends the first half of its $2m$ qubits to a receiver, Bob, and the other half to another receiver, Charlie. Suppose further that Alice reveals her choice of $W$ to the receivers and asks Bob to guess $x$ and Charlie to guess $z$.
They can only communicate before the transmission starts. The concern of this paper is the probability that Bob and Charlie make correct guesses simultaneously. (Any $x'$ and $z'$ s.t.\  $x+x'\in W$ and $z+z'\in W^\perp$ constitute a correct answer.)

The authors of \cite{coladangelo2021hidden} conjectured that the probability of Bob's and Charlie's guesses being simultaneously correct goes to zero exponentially fast as with $n$, even if they can request that Alice passes $\ket{W_{x,z}}$ through a quantum channel (they chose) before revealing $W$. They showed that the validity of the conjecture implies the security of several cryptographic protocols. The proof of the conjecture appeared later in \cite{CV22}. The authors derived an exponentially decreasing upper bound on the simultaneous correct guessing probability by using
various general operator inequalities and techniques used in \cite{tomamichel2013monogamy} to show the monogamy of entanglement in some related problems. There are no claims in \cite{CV22} that the bound is optimal nor methods to construct Bob and Charlie's strategies to achieve it.

We focus on the scenario where Bob and Charlie cannot request that Alice pass $\ket{W_{x,z}}$ through some quantum channel before revealing $W$. We first derive an upper bound lower than the one derived in \cite{CV22}, then show how Bob and Charlie can achieve it. Thus, the bound is optimal in this scenario. The proof techniques rely on the algebra and combinatorics of the vector space $\F_2$, different from those used in \cite{coladangelo2021hidden, tomamichel2013monogamy}. 

We present an optimal guessing strategy for Bob and Charlie by using and modifying some ideas from (quantum CSS) coding theory. 
\rev{The authors of \cite{culf2022group} also observed a connection between coset monogamy games with CSS codes but did not explore strategies for these games.}
\rev{While an encoder for CSS code states using Hadamard and $\CNOT$ gates is established \cite{Lai2017},
our optimal strategy includes a novel and explicit way of mapping an all-zero qubit state to EPR pairs shared between Bob and Charlie, then to the corresponding CSS-encoded state also using only Hadamard and $\CNOT$ gates. Our encoding map also naturally leads to a simultaneous optimal decoder for both Bob and Charlie, which is essential to the optimal strategy. 
}
This map is an offshoot contribution of possible interest to quantum coding theorists as simplifying circuitry for encoding, decoding, and operating systems implementing CSS codes is an active subject of current research \cite{CSS:RengaswamyCNP20}.

We do not intend to claim that our results have further cryptographic relevance. \rev{Beyond cryptography, quantum games, and the resulting correlations that arise, are tools for studying problems in physics foundations. For example, they provide a natural setting for the experimental measurement of Bell inequalities \cite{brunner2014bell}.} Our interests come from the information theory of distributed communications and computing systems. Here, we aim to understand and quantify the classical and quantum correlations between Bob's and Charlie's guesses. To explain what this paper accomplishes toward that goal, consider the following classical scenario.

Suppose Alice reveals to Bob and Charlie her choice of $W\subset \mathbb{F}_2^{2m}$ but withholds information about the $x$ and $z$ she picked uniformly at random. She asks Bob to guess the $x$-coset of $W$ and Charlie the $z$-coset of $W^{\perp}$ without sending them any further classical or quantum information. The probability that Bob guesses correctly is $1/2^m$, which is also true for Charlie.
The probability that both independently and simultaneously guess correctly is thus $1/2^{2m}$. This paper shows that when Alice splits the qubits of  $\ket{W_{x,z}}$ and sends half to Bob and half to Charlie, the optimal probability of simultaneous correct guessing increases to $\Theta(1/2^{m})$ while the individual correct guessing probabilities on average remain on the order of $1/2^m$. \rev{ The primary connection we demonstrate to quantum error correction is the analysis of simultaneous optimal decoding for Bob and Charlie when given parts of a coset state. Namely, given their respective halves of a coset state, the best decoder Bob and Charlie can jointly use to decode the entire state is comprised of optimal decoders on both sides. 
}



The rest of the paper is organized as follows: In \Cref{sec:game}, we review some notation and basic definitions and formally define coset guessing as a game. In \Cref{sec:upperbound}, we derive an upper bound on Bob's marginal probability of winning. In \Cref{sec:cosetEncoding}, we develop a canonical method of encoding subspace states using Hadamard and $\CNOT$ gates.
Finally, in \Cref{sec:optstrat}, we present a strategy that  Bob and Charlie can use to achieve the bound.

\section{The Coset Guessing Game \label{sec:game}}

\subsection{Preliminaries}\label{subsec:prelims}
Here, we state the notation and define the mathematical objects we need to describe the game. We need mostly standard algebraic and basic quantum computing notions provided in the appendix. For a quick primer on quantum computing, we refer the reader to \cite{QC:Soljanin20}.
\\[2ex]
\ul{Notation}
\begin{itemize}[leftmargin=0.5cm]
\item $\ket{\phi} =  \CNOT_{i_1, j_1}H_{i_1}\cdots\CNOT_{i_\ell, j_\ell}H_{i_\ell}\ket{y}$ and 
    \item $\F_2^{k}$ -- The vector space of dimension $k$ over the field $\F_2$.
    \item $G(n,k)$ --  The set of all subspaces $W\subseteq \F^n_2$ of dimension $0 \leq k\leq n$.
    \item $W^\perp$ -- The dual space of $W$, i.e., the set of all $z \in \F_{2}^n$ for which $x\cdot z = 0$ for all $x \in W$.
    \item $x_i^j$ -- The $i,i+1,\ldots, j$-th coordinates of a vector $x \in \F_2^n$.
    \item $\binom{n}{k}_q$ -- The q-binomial coefficient, equal to $\prod_{j=0}^{k-1} \frac{q^{n-j}-1}{q^{k-j}-1}$, and represents the cardinality of $G(n,k)$ for $q=2$.
    \item $\CS(W)$ -- A set of coset representatives for subspace $W\subseteq \bbF_2^n$.
    \item $\CNOT_{i,T}$ -- The set of $\CNOT$ gates acting on an $n$-qubit register where $i$ is the control bit index and $T$ is the set the target bit indices.
    \item $[n] = \{1,2,\ldots, n\}$ and $[k,n] = \{k,k+1,\ldots, n\}$.
    \item \rev{$\supp{\vec v}$ - Coordinates where $\vec v$ has non-zero values.}
\end{itemize}

\vspace{1ex}
\noindent\ul{Quantum Coset States}\\[1ex]
For a subspace $W \in G(2m,m)$, the \textit{subspace state} $\ket{W} \in (\C^2)^{\otimes 2m}$ is a uniform superposition of computational basis vectors with labels in $W$:
\begin{equation}
    \ket{W} = \frac{1}{\sqrt{2^m}} \sum_{u \in W} \ket{u}
\end{equation}

Given vectors $x,z \in \bbF_2^{2m}$, we define the \textit{coset state} $\ket{W_{x,z}}\in (\C^2)^{\otimes 2m}$ by
\begin{equation}
    \ket{W_{x,z}} = \frac{1}{\sqrt{2^m}} \sum_{u \in W} (-1)^{z \cdot u} \ket{x + u}.
\end{equation} 
The state $\ket{W_{x,z}}$  depends on $W$, the coset $W + x$, and the coset $W^\perp + z$. Specifically, if $x$, $x'$ belong to the same coset of $W$, and $z$, $z'$ belong to the same coset of $W^{\perp}$, then $\ket{W_{x,z}} = \pm \ket{W_{x',z'}}$, meaning they are equal up to some global phase factor. Thus, for each subspace $W$, we can generate all coset states up to a global phase by iterating through a choice of $\CS(W)$ and $\CS(W^\perp)$.

\begin{example}
    \label{ex:CosetStateIntro}
    Consider dimension one subspaces of $\F_2^2$ There are three subspaces of dimension $1$ in $\F^2_2$.
\[
 W^{(1)}=\left\{\bigl[\begin{smallmatrix}
    0\\
    0
\end{smallmatrix}\bigr],\bigl[\begin{smallmatrix}
    0\\
    1
\end{smallmatrix}\bigr]  \right\} ~~ \!\!
W^{(2)}=\left\{\bigl[\begin{smallmatrix}
    0\\
    0
\end{smallmatrix}\bigr],\bigl[\begin{smallmatrix}
    1\\
    0
\end{smallmatrix}\bigr]  \right\} ~~ \! \!
W^{(3)}=\left\{\bigl[\begin{smallmatrix}
    0\\
    0
\end{smallmatrix}\bigr],\bigl[\begin{smallmatrix}
    1\\
    1
\end{smallmatrix}\bigr]  \right\} \!\!
\]
Each subspace, with its cosets and the cosets of its dual, defines four different coset states as displayed in Figure \ref{fig:CosetStatesExample}.
{
\setlength{\textfloatsep}{0pt}
\begin{figure*}[!t]
    \captionsetup{aboveskip=0pt, belowskip=0pt}
    \centering
    \label{fig:CosetStatesExample}
    \begin{minipage}{\textwidth}
    \begin{small}
	\begin{align*}
    \CS(W^{(1)}) = \{00, 10\} & & \CS(W^{(2)}) = \{00, 01\} & & \CS(W^{(3)}) = \{00, 01\}\\
    \CS(W^{(1)\perp}) = \{00, 01\} & & \CS(W^{(2)\perp}) = \{00, 10\} & & \CS(W^{(3)\perp}) = \{00, 10\}\\
    \ket{W_{00,00}^{(1)}} = \ket{0} \otimes \ket{+} & & \ket{W_{00,00}^{(2)}} = \ket{+} \otimes \ket{0} & & \ket{W_{00,00}^{(3)}} = \frac{1}{\sqrt{2}} \left( \ket{00} + \ket{11} \right)\\
    \ket{W_{00,01}^{(1)}} = \ket{0} \otimes \ket{-} & & \ket{W_{00,10}^{(2)}} = \ket{-} \otimes \ket{0} & & \ket{W_{00,10}^{(3)}} = \frac{1}{\sqrt{2}} \left( \ket{00} - \ket{11} \right)\\
    \ket{W_{10,00}^{(1)}} = \ket{1} \otimes \ket{+} & & \ket{W_{01,00}^{(2)}} = \ket{+} \otimes \ket{1} & & \ket{W_{01,00}^{(3)}} = \frac{1}{\sqrt{2}} \left( \ket{01} + \ket{10} \right)\\
    \ket{W_{10,01}^{(1)}} = \ket{1} \otimes \ket{-} & & \ket{W_{01,10}^{(2)}} = \ket{-} \otimes \ket{1} & & \ket{W_{01,10}^{(3)}} = \frac{1}{\sqrt{2}} \left( \ket{01} - \ket{10} \right).
\end{align*}
\end{small}
\caption{Examples of coset states for $m = 1$.}
    \vspace{-\baselineskip}
    \end{minipage}

\end{figure*}
}
For a general $W \in G(2m, m)$, there will be $2^{2m}$ different coset states.
\end{example}

\subsection{The Game Dynamics}
The game has three players: Alice, Bob, and Charlie.
Alice challenges Bob and Charlie by sending a $2m$-qubit coset state $\ket{W_{x,z}}$ based on a random choice of $W\in G(2m,m)$, $x\in \mathbb{F}_2^{2m}$, and $z\in \mathbb{F}_2^{2m}$. She sends the first $m$ qubits of $\ket{W_{x,z}}$ to Bob and the rest to Charlie and then reveals her choice of $W$. 

Bob's task is to guess some $x^\prime\in W+x$, and Charlie's is to guess a $z^\prime\in W^{\perp} + z$. Thus, given a $W \in G(2m,m)$, Bob uses a POVM $\{B_x^W\}_{x \in \CS(W)}$ and Charlie uses a POVM $\{C_z^W\}_{z \in \CS(W^{\perp})}$.
They locally measure their qubits, or equivalently, they measure $\ket{W_{x,z}}$ with the product measurements $\{B_{x'}^W \otimes C_{z'}^W\}_{x',z'}$ to get outcomes $x' \in \CS(W)$ and $z' \in \CS(W^{\perp})$. Thus, they win with probability $\Tr\bigl[(B_x^W \otimes C_z^W) \ket{W_{x,z}}\bra{W_{x,z}}\bigr]$.
Their expected \ul{winning probability} is the average over all possible Alice's choices of $W \in G(2m,m)$, $ x \in \CS(W)$,  $z \in \CS(W^\perp)$:
\begin{equation}
    p_m= {E} \Bigl\{ \Tr\bigl[(B_x^W \otimes C_z^W) \ket{W_{x,z}}\bra{W_{x,z}}\bigr]\Bigr\},
\label{eq:wp}
\end{equation}
where $E$ denotes the expectation (average) over $W \in G(2m,m)$, $ x \in \CS(W)$,  $z \in \CS(W^\perp)$ picked uniformly at random. The game dynamics are illustrated in Figure \ref{fig:game_diagram}.
\vspace{3pt}
\begin{algorithm}[h]
\label{alg:game}
\caption{\textbf{Coset Guessing Game}}
\begin{algorithmic}
\STATE \textbf{Parameters:}
A positive integer $m \in \Z$.
\STATE \textbf{Strategy:}\\
For each subspace $W \in G(2m,m)$, Bob and Charlie select POVMs $\{B_x^W\}_{x \in \CS(W)}$ and $\{C_z^W\}_{z \in \CS(W^{\perp})}$, respectively.\\[0.5ex]
\textbf{Game Procedure:}
\begin{enumerate}
    \item Alice samples $W \in G(2m,m),x \in \CS(W),$ and $z \in \CS(W^{\perp})$ uniformly at random. She prepares the state $\ket{W_{x,z}}$, comprised of $2m$ qubits.
    \item Alice sends her state $\ket{W_{x,z}}$ to Bob and Charlie so that Bob receives the first $m$ qubits and Charlie receives the last $m$ qubits.
    \item Alice announces to Bob and Charlie which subspace $W \in G(2m,m)$ she picked.
    \item Bob and Charlie locally measure $\ket{W_{x,z}}$ with the product measurements $\{B_{x'}^W \otimes C_{z'}^W\}_{x',z'}$ to get outcomes $x' \in \CS(W)$ and $z' \in \CS(W^{\perp})$. 
\end{enumerate}
\STATE \textbf{Win Condition:} $x+x' \in W$ and $z+z' \in W^{\perp}$.
\end{algorithmic}
\end{algorithm}

\begin{example}
\label{ex:N1GameIntro}
    Recall in \Cref{ex:CosetStateIntro}, when $m = 1$, there are 12 different coset states.
In this simple example, we immediately see the problem that arises when Bob and Charlie try to guess $x$ and $z$. If $W = W^{(1)}$, Bob can always guess $x$ by simply measuring his qubit in the computational basis, $\{\ket 0, \ket{1}\}$. Moreover, Charlie can always guess $z$ by measuring his qubit in the Hadamard basis $\{\ket +, \ket{-}\}$. This is not true when $W=W^{(2)}$ since $x$ affects Charlie's qubit, and $z$ affects Bob's qubit. In the case $W = W^{(3)}$, there is entanglement between Bob and Charlie that makes things more interesting. We will return to this example in \Cref{sec:optstrat}.
\end{example}

\begin{figure}[H]
    \centering
    \begin{tikzpicture}[scale=1.06]
        \node[] at (0, 2.3) {Choose $W \in G(2m, m), x \in \CS(W)$, and $z \in \CS(W^\perp)$.};
        \node[] at (0, 1.5) {Prepare $\ket{W_{x, z}}$.};
        \node[alice,minimum size=1.5cm] (A) at (0,0) {Alice};
        \node[bob,minimum size=1.5cm] (B) at (-3,-3) {Bob};
        \node[charlie, mirrored, minimum size=1.5cm] (C) at (3,-3) {Charlie};
        \node[text width=2.5cm] at (2.9, -1.2) {\small Second $m$ qubits of $\ket{W_{x, z}}$.};
        \node[text width=2.1cm] at (-2.3, -1.2) {\small First $m$ qubits of $\ket{W_{x, z}}$.};
        \node[] at (3, -4.7) {Goal: output $z$.};
        \node[] at (-3, -4.7) {Goal: output $x$.};

        \draw[->] (A) -- (B);
        \draw[->] (A) -- (C);
    \end{tikzpicture}
    \caption{An illustration of the considered coset guessing game.}
    \label{fig:game_diagram}
\end{figure}

This game is a variant of the game in \cite{CV22}. There, Bob and Charlie can request that Alice perform a transformation $\Phi(\ket{W_{x,z}}\bra{W_{x,z}})$ before splitting and sending qubits.
For convenience, the notation used in this paper is slightly different than that used in \cite{CV22}. We write $G(2m, m)$ instead of $G(n, \frac{n}{2})$ and also use characters that are standard in coding theory. Finally, we also note that generalization of our results in Sections \ref{sec:upperbound}-\ref{sec:optstrat} to arbitrary $n$ and $k$ (where Alice randomly samples a subspace $W \in G(n,k)$ and $x,z \in \F_2^n$, sending $k$ qubits of the state $\ket{W_{x,z}}$ to Bob and $n-k$ qubits to Charlie) is straightforward.

\section{An Upper Bound on the Winning Probability}\label{sec:upperbound}
\begin{theorem}
    \label{thm:pnupperbound}
    Let $p_m$ denote the supremum over all strategies of the winning probability for the coset guessing game parameterized by $m \in \N$. The following upper bound holds: 
    \[p_m \le \frac{1}{\binom{2m}{m}_2}\sum_{k=0}^m 2^{k^2} \binom{m}{k}_2^2 \left(\frac{1}{2}\right)^k = \Theta \left ( \frac{1}{2^m} \right ). \]
\end{theorem}
\begin{proof}
    Fix a subspace $W \in G(2m,m)$. We aim to show that the maximal probability of winning conditioned on the subspace $W$ being chosen is at most $\frac{1}{2^{m-k}}$, where $k=\dim(W \cap \langle e_i \rangle_{i=m+1}^{2m})$, from which the upper bound will follow by counting subspaces.

    Given $W,x$, and $z$, the selection of POVMs $\{B_x^W\}_{x \in \CS(W)}$ and $\{C_z^W\}_{z \in \CS(W^{\perp})}$ induces a probability distribution for Bob and Charlie's guesses $\widehat{x}$ and $\widehat{z}$:

    \begin{equation}
        p(\widehat{x},\widehat{z}|W,x,z) = \Tr\bigl[(B_{\widehat{x}}^W \otimes C_{\widehat{z}}^W) \ket{W_{x,z}}\bra{W_{x,z}}\bigr]
    \end{equation}
    The winning probability conditioned on the choice of subspace and coset representatives is thus $p(x,z|W,x,z)$, which  is upper bounded by Bob's marginal winning probability
    \begin{align}
        p(x|W,x,z) &= \sum_{\widehat{z}} p(x,\widehat{z}|W,x,z) \nonumber\\
        &=  \Tr\bigl[\bigl(B_x^W \otimes \bigl( \sum_{\widehat{z}}C_{\widehat{z}}^W \bigr)\bigr) \ket{W_{x,z}}\bra{W_{x,z}}\bigr] \nonumber\\
        &= \Tr\bigl[\bigl(B_x^W \otimes Id_C \bigr) \ket{W_{x,z}}\bra{W_{x,z}}\bigr] \nonumber\\
        &= \Tr\bigl[B_x^W \Tr_C\bigl[\ket{W_{x,z}}\bra{W_{x,z}}\bigr]\bigr]
    \end{align}
Here, we used the fact that POVM elements sum to identity and properties of the partial trace $\Tr_C$ over Charlie's $m$ qubits. Averaging over $x$ and $z$, we  determine Bob's marginal winning probability conditioned solely on the choice of subspace $W$:
    \begin{align}\label{eq:qsd}
        p^B_{\text{win}}(W) &= \sum_{x,z} p(x,z)p(x|W,x,z) \nonumber\\
        &= \frac{1}{4^m} \sum_{x,z} \Tr\bigl[B_x^W \Tr_C\bigl[\ket{W_{x,z}}\bra{W_{x,z}}\bigr]\bigr] \nonumber\\
        &= \frac{1}{2^m} \sum_x  \Tr\bigl[B_x^W \rho^B_x\bigr]
    \end{align}
    where $\rho^B_x = \frac{1}{2^m}\sum_z \Tr_C\bigl[\ket{W_{x,z}}\bra{W_{x,z}} \bigr]$. The last line of equation \eqref{eq:qsd} is now a problem of discriminating the quantum states $\{\rho_x^B\}_x$ with uniform probability for Bob. In general, optimization of quantum state discrimination requires semidefinite programming to solve. However, it turns out that many of these states end up being identical or orthogonal to each other, as shown in \Cref{lma:tracestates}:
    \begin{lemma}\label{lma:tracestates}
    For every $x, y\in \CS(W)$, the states $\rho_x^B$ and $\rho_{y}^B$ are either identical or orthogonal to each other, with equality occurring exactly when $x+y \in W + \langle e_i \rangle_{i=m+1}^{2m}$, i.e. they belong to the same coset of $W + \langle e_i \rangle_{i=m+1}^{2m}$.
    \end{lemma}
    \textit{Proof:} Let $C_u$ denote the equivalence class of vectors $v\in W$ satisfying $v_{m+1}^{2m}=u_{m+1}^{2m}$. Two vectors $u,v \in W$ lie in the same equivalence class iff $u+v \in \langle e_i\rangle_{i=1}^m$, so the number of equivalence classes is $2^r$, each with size $2^{m-r}$, where $r= m - \dim(W \cap \langle e_i \rangle_{i=1}^{m})$. Define the state $\ket{\psi_{u,z}^x} = \frac{1}{\sqrt{2^{m-r}}} \sum_{v \in C_u}(-1)^{z_1^m \cdot v_1^m} \ket{x_1^m + v_1^m}$ for an equivalence class $C_u$ and $z \in \CS(W^{\perp})$. Then, 

    \begin{align}
        \Tr_C&\bigl[\ket{W_{x,z}}\bra{W_{x,z}}\bigr]\bigr] \nonumber\\
        &= \frac{1}{2^m}\sum_{v,w \in W} (-1)^{z \cdot (v+w)} \Tr_C\bigl[ \ket{x+v}\bra{x+w} \bigr] \nonumber\\
        &= \frac{1}{2^m} \sum_{C_u}\sum_{v,w \in C_u} (-1)^{z \cdot (v+w)} \ket{x_1^n+v_1^n}\bra{x_1^n+w_1^n}  \nonumber\\
        &= \frac{1}{2^m} \sum_{C_u}\sum_{v,w \in C_u} (-1)^{z_1^m \cdot (v_1^m+w_1^m)} \ket{x_1^n+v_1^n}\bra{x_1^n+w_1^n}  \nonumber \\
        &= \frac{1}{2^r} \sum_{C_u} \ket{\psi_{u,z}^x} \bra{\psi_{u,z}^x}
    \end{align}
    Thus $\rho_x^B = \frac{1}{2^{m+r}} \sum_z \sum_{C_u} \ket{\psi_{u,z}^x} \bra{\psi_{u,z}^x}$.  Suppose that $x + y \not \in  W + \langle e_i \rangle_{i=m+1}^{2m}$. Then, there exist no vectors $v,w \in W$ for which $x_1^m + v_1^m = y_1^m+w_1^m$, so the vectors $\ket{\psi_{u,z}^x}$ and $\ket{\psi_{u',z'}^y}$ are orthogonal due to their disjoint basis support. It follows immediately that $\rho_x^B$ and $\rho_{y}^B$ are orthogonal.

    Now, suppose that $x + y \in  W + \langle e_i \rangle_{i=m+1}^{2m}$. Let $w \in W$ be chosen so that $x_1^m + y_1^m = w_1^m$, ensuring that $(w+v)_1^m$ iterates over $\{v_1^m : v \in C_u\}$ as $v$ iterates over $C_u$. Therefore, we have
    \begin{align}
        \ket{\psi_{u,z}^x} &= \frac{1}{\sqrt{2^{m-r}}} \sum_{v \in C_u}(-1)^{z_1^m \cdot v_1^m} \ket{x_1^m + v_1^m} \nonumber \\
        &= \frac{1}{\sqrt{2^{m-r}}} \sum_{v \in C_u}(-1)^{z_1^m \cdot v_1^m} \ket{y_1^m + w_1^m + v_1^m} \nonumber\\
        &= \frac{(-1)^{z_1^m \cdot w_1^m}}{\sqrt{2^{m-r}}} \sum_{v \in C_u}(-1)^{z_1^m \cdot (w+v)_1^m} \ket{y_1^m + (w+v)_1^m}\nonumber \\
        &= (-1)^{z_1^m \cdot w_1^m} \ket{\psi_{u,z}^y}
    \end{align}

    Thus $\rho_x^B = \frac{1}{2^{m+r}} \sum_z \sum_{C_u} \ket{\psi_{u,z}^x} \bra{\psi_{u,z}^x} = \frac{1}{2^{m+r}} \sum_z \sum_{C_u} \ket{\psi_{u,z}^y} \bra{\psi_{u,z}^y} = \rho_y^B$.

    Now, recalling $k=\dim(W \cap \langle e_i \rangle_{i=m+1}^{2m})$ and noting that any coset $C$ of $W + \langle e_i \rangle_{i=m+1}^{2m}$ contains exactly $2^{m-k}$ cosets of $W$, Bob cannot distinguish the states $\rho_x^B,\rho_y^B$ by any measurement for any coset representatives $x,y$ that belong to the same coset of $W + \langle e_i \rangle_{i=m+1}^{2m}$. Denoting the coset of $W + \langle e_i \rangle_{i=m+1}^{2m}$ containing $x$ by $\mathcal{C}_x$, we have
    \begin{align}
        p^B_{win}(W) &= \frac{1}{2^m} \sum_x  \Tr\bigl[B_x^W \rho^B_x\bigr]
        = \frac{1}{2^m} \sum_{\mathcal{C}_x} \sum_{y \in \mathcal{C}_x}  \Tr\bigl[B_y^W \rho^B_x\bigr]\nonumber\\
        &= \frac{1}{2^m} \sum_{\mathcal{C}_x}  \Tr\bigl[ \bigl( \sum_{y \in \mathcal{C}_x}B_y^W \bigr) \rho^B_x\bigr] \nonumber\\
        & \leq \frac{1}{2^m} \sum_{\mathcal{C}_x}  \Tr\bigl[\rho^B_x\bigr]
        = \frac{2^k}{2^m} = \frac{1}{2^{m-k}}
    \end{align}

    Therefore, we can upper bound Bob's marginal probability of winning for the subspace $W$, in terms of the dimension of the intersection of $W$ and $\langle e_i \rangle_{i=m+1}^{2m}$. We note that due to the orthogonality relations determined in \Cref{lma:tracestates}, Bob can achieve this winning probability with any POVM satisfying $\sum_{y \in C_x} B_y^W = 2^{m-k}\rho_x^B$ for every equivalence class $C_x$.

    Now it is a matter of counting the subspaces $W \in G(2m,m)$ with $\dim (W \cap \langle e_i \rangle_{i=m+1}^{2m}) = k$ for each $0 \leq k \leq m$. Results on enumerating subspaces over finite fields can be found throughout the literature, and we refer the reader to \cite{braun2018q} from which the following result is derived.

    \begin{lemma}
        There are $2^{(m-k)^2}\binom{m}{m-k}^2_2$ subspaces $W \in G(2m,m)$ whose intersection with $\langle e_i \rangle_{i=m+1}^{2m}$ has dimension $k$, where $\binom{m}{k}_2$ is the q-binomial coefficient.
    \end{lemma}

    \textit{Proof:} Fix a subspace $V \subseteq \langle e_i \rangle_{i=m+1}^{2m}$ of dimension $k$. Using \cite[Lemma 1]{braun2018q} on the chain of subspaces $V \subseteq \langle e_i \rangle_{i=m+1}^{2m} \subseteq \F_2^{2m}$, there are $2^{(m-k)^2}\binom{m}{m-k}_2$ subspaces $W \subseteq \F_2^{2m}$ of dimension $n$ satisfying $W\cap \langle e_i \rangle_{i=m+1}^{2m} = V$. Recalling that there are $\binom{m}{k}_2=\binom{m}{m-k}_2$ subspaces $V\subseteq  \langle e_i \rangle_{i=m+1}^{2m}$ of dimension $m$ proves the lemma.
    
    We express the upper bound for the winning probability as
    \begin{align}
        p_m &\leq \frac{1}{\binom{2m}{m}_2}\sum_{k=0}^m 2^{(m-k)^2}\binom{m}{m-k}_2^2 \left(\frac{1}{2} \right)^{m-k} \nonumber\\
        &= \frac{1}{\binom{2m}{m}_2}\sum_{k=0}^m 2^{k^2}\binom{m}{k}_2^2 \left( \frac{1}{2} \right)^k
    \end{align}

    \Cref{lma:rate} describes the rate of decay for the upper bound:

    \begin{lemma}\label{lma:rate}
    The asymptotic behavior of the bound satisfies $\frac{1}{\binom{2m}{m}_2}\sum_{k=0}^m 2^{k^2}\binom{m}{k}_2^2 \left(\frac{1}{2} \right)^{k} = \Theta\bigl(\frac{1}{2^m}\bigr)$.
    \end{lemma}
    \textit{Proof:} Let $f(m,k)=\frac{2^{k^2}\binom{m}{k}_2^2}{\binom{2m}{m}_2}$. The q-analog of Vandermonde's identity states that \[\binom{2m}{m}_2 = \sum_{k=0}^m 2^{k^2}\binom{m}{k}_2^2\] and so $f(m,k) \leq 1$ for all $0 \leq k \leq m$. By computing the ratio $\frac{f(m,k+1)}{f(m,k)}$ for $0 \leq k \leq m-2$, we observe that

    \begin{align}
        \frac{f(m,k+1)}{f(m,k)} &= \frac{2^{(k+1)^2}\binom{m}{k+1}_2^2}{2^{k^2}\binom{m}{k}_2^2} \nonumber\\
        &= 2^{(k+1)^2 - k^2} \left( \frac{\prod_{i=0}^k \frac{2^{m-i}-1}{2^{k+1-i}-1}}{\prod_{i=0}^{k-1} \frac{2^{m-i}-1}{2^{k-i}-1}} \right)^2 \nonumber\\
        &= 2^{2k+1}\left( \frac{2^{m-k} - 1}{2^{k+1}-1} \right)^2 \nonumber\\
        &\geq \frac{1}{2}(2^{m-k}-1)^2 \geq \frac{9}{2}.
    \end{align}

    It follows inductively that $f(m,k) \leq (\frac{2}{9})^{m-k-1}f(m,m-1)$ for $0\leq k \leq m-2$, so we have

    \begin{align}\label{eq:bigO}
        \frac{1}{\binom{2m}{m}_2}&\sum_{k=0}^m 2^{k^2}\binom{m}{k}_2^2 \left(\frac{1}{2} \right)^{k} =\sum_{k=0}^m f(m,k) \left(\frac{1}{2} \right)^{k} \nonumber\\
        &\leq \sum_{k=0}^{m-1} f(m,m-1)  \left(\frac{2}{9} \right)^{m-k-1}\left(\frac{1}{2} \right)^{k} + \frac{f(m,m)}{2^m} \nonumber\\
        &\leq \left( \frac{2}{9}\right)^{m-1}\sum_{k=0}^{m-1}  \left(\frac{9}{4} \right)^{k}+ \frac{1}{2^m} \nonumber\\
        &= \left( \frac{2}{9}\right)^{m-1} \frac{\left( \frac{9}{4}\right)^{m} - 1}{\frac{9}{4}-1} + \frac{1}{2^m} \nonumber\\
        & \leq \frac{9}{2}\left(\frac{1}{2^m} \right) + \frac{1}{2^m} = \frac{11}{2}\left(\frac{1}{2^m} \right) = O(\frac{1}{2^m}).
    \end{align}

 By the q-analog of Vandermonde's identity \[\frac{1}{\binom{2m}{m}_2}\sum_{k=0}^m f(m,k)\left(\frac{1}{2}\right)^k \!\!\! \geq \frac{1}{\binom{2m}{m}_2}\sum_{k=0}^m f(m,k)\left(\frac{1}{2}\right)^m \! \! = \frac{1}{2^m}.\]
 
 Thus, the upper bound asymptotically satisfies $\frac{1}{\binom{2m}{m}_2}\sum_{k=0}^m 2^{k^2}\binom{m}{k}_2^2 \left(\frac{1}{2} \right)^{k} = \Theta(\frac{1}{2^m})$. We remark that the upper bound constant of $\frac{11}{2}$ in \eqref{eq:bigO} can be made smaller by a more careful derivation.
\end{proof}

\section{Coset State Encoding and Representation}
\label{sec:cosetEncoding}
The rest of this paper proves the bound in \Cref{thm:pnupperbound} is achievable. In this section, we first provide an efficient circuit for constructing coset states, and in Section \ref{sec:optstrat}, we present an optimal winning strategy based on the derived representation.
In the next subsection, we devise a $\CNOT$ and $H$ gates circuit $N_W$ to encode $\ket{0}^{\otimes 2m}$ into $\ket{W}$. We then show that we can pick coset representatives $x$ \& $z$ such that $\supp(x) \cap \supp(z)=\emptyset$ and
\[
\ket{W_{x, z}}  
				 = N_W\ket{x+z}.
\]
The main idea of having this representation is to enable Bob and Charlie to invert $N_W$ by local actions and then measure to obtain the bits of $x + z$, which is only partially doable when $N_W$ contains nonlocal actions. The details are in \Cref{sec:optstrat}.

These ideas are, to some extent, coding-theoretic.
The reader familiar with CSS codes may recognize that the CSS$(C_1,C_2)$ with $C_1=\mathbb{F}_2^{2m}$ and $C_2=W$ could see
1) $\ket{W}$  as the image of $\ket{0}^{\otimes 2m}$ and 2) the coset states as the error states $X$ and $Z$: $\ket{W_{x, z}} =X_{\supp(x)}Z_{\supp(z)}\ket{W}$. Under this interpretation, Bob's task would be to decode $X$ errors and Charlie's would be to decode $ Z$ errors\rev{, each using only local measurements}. For global decoding, a coding theorist would pick minimum-weight coset representatives $x$ and $z$ for maximum likelihood decoding. \rev{While our decoder does not have a high likelihood of accurate decoding, it is comprised of entirely local operations and is the optimal decoder that does so.}


\subsection{Subspace State Encoding}
\label{subsec:ConstructingCosetState}

In this subsection, we devise a $\CNOT$ and $H$ gates circuit $N_W$ to map (encode) the state $\ket{0}^{\otimes 2m}$ into a subspace state $\ket{W}$. The main result is shown in \Cref{prop:ketWGoodRep}. We explain the idea by considering a subspace state corresponding to a subspace of $\F_2^{2m}$ with a single generator $v_i$. 

Let $v_i \in \F_2^{2m}$ be a binary string of length $m$, whose first nonzero entry lies at the index $i$. Let $S_i$ denote the support of the vector $v_i$, and let $B_i = S_i \backslash \{i\}$. Then the circuit $\CNOT_{i,B_i}H_i$ (where the index $i$ is the control bit and the indices $j \in B_i$ are the target bits) performs the following mapping:
\[
\ket{0}^{\otimes 2m} \overset{\CNOT_{i,B_i}H_i}{\mapsto} \frac{1}{\sqrt{2}} \bigl( \ket{0}^{\otimes 2m} + \ket{v_i} \bigr)
\]
Moreover, given any vector $c \in \F_2^{2m}$ whose $i$-th bit is zero, the unitary transformation $\CNOT_{i,B_i}H_i$ maps the vector 
\begin{equation}
    \label{eq:generateSuperPosition}
    \ket{c} \mapsto \frac{1}{\sqrt{2}} \bigl( \ket{c} + \ket{c+v_i} \bigr).
\end{equation} Thus, given another vector $v_j \in \F_2^{2m}$ whose first non-zero entry lies at index $j$, we can generate a uniform superposition of the subspace generated by $v_i$ and $v_j$ via
\begin{align*}
    \ket{0}^{\otimes 2m} &\overset{\CNOT_{i,B_i}H_i}{\mapsto} \frac{1}{\sqrt{2}} \bigl( \ket{0}^{\otimes 2m} + \ket{v_i} \bigr)\\ &\overset{\CNOT_{j,B_j}H_j}
    {\mapsto} \frac{1}{2} \bigl( \ket{0}^{\otimes 2m} + \ket{v_j} + \ket{v_i} + \ket{v_i+v_j} \bigr)
\end{align*}
as long as the $i$-th bit of $v_j$ is zero. We can continue this process with any number of vectors to generate the subspace spanned by those vectors.


We will use this method to generate $\ket{W}$, the uniform superposition of the computational basis vectors with labels in $W$. Let $A_W$ be the matrix whose rows generate $W$. We will consider $A_W$ in reduced-row echelon form and refer to it as the generator matrix of $W$. Since $A_W$ is in reduced row echelon form, the index of the first $1$ appearing in the bit string of a row vector $A_W$ is a pivot column. This means this bit is zero for all other row vectors, so we can apply \Cref{eq:generateSuperPosition} when needed.
We are now ready to describe the circuit that maps $\ket{0}^{\otimes 2m}$ into $\ket{W}$.

\begin{proposition}
    \label{prop:ketWGoodRep}
     Given $W \in G(2m, m)$, let $A_W$ be the generator matrix of $W$ in reduced row echelon form. Let $I$ denote the indices of the pivot columns of $A_W$ and let $J=\{(i,j):i \in I, j \in I^c, A_{ij}=1\}$. Then
    \[\ket W = \CNOT_J H_I \ket{0}^{\otimes 2m}.\]
\end{proposition}
\begin{proof}
    The first entry of a pair in $J$ is always a pivot. For a fixed pivot $i_0 \in I$, there is exactly one row-vector $a_{i_0}$ in $A_W$ with a $1$ at index $i_0$. The second entries in $J$, $\{j : (i_0, j) \in J\}$, form the support of the row vector $a_{i_0}$, excluding the pivot $i_0$. We can therefore decompose $\CNOT_J$ as a product in terms of $B_i = \supp(a_i) \setminus \{i\}$: \[\prod_{i \in I} \CNOT_{i, B_i}.\]
    The targets of these $\CNOT$ gates never contain an element of $I$, and each control bit is in $I$. Thus, for a specific $i_0 \in I$ we can commute $H_{i_0}$ with all gates except for $\CNOT_{i_0, B_{i_0}}$. This allows us to write
    \[\CNOT_J H_I \ket{0}^{\otimes 2m} = \prod_{i \in I} \CNOT_{i, B_i} H_i \ket{0}^{\otimes 2m}.\]
    Using \Cref{eq:generateSuperPosition}, we end up with a superposition of the subspace generated by all $a_i$, which is exactly $\ket{W}$.
\end{proof}

\subsection{Choosing Coset Representatives}
\rev{
Coset representatives in classical coding theory are minimum weight words since they correspond to most likely errors on memoryless channels. Here, different properties of coset representatives are helpful.}
\Cref{prop:ketWGoodRep} explains how we can see Alice's preparation of $\ket{W}$ as passing the $\ket{0}^{\otimes 2m}$ state through a circuit of depth at most $m + 1$ to create $\ket{W}$. The circuit consists of Hadamard gates followed by $\CNOT$ gates. Alice can thus prepare the coset state
as
\begin{align}
 \label{eq:csN}   
\ket{W_{x,z}} &=  X_{\supp(x)}Z_{\supp(z)}\ket{W} \nonumber\\
&=X_{\supp(x)}Z_{\supp(z)}\CNOT_J H_I \ket{0} ^{\otimes 2m} 
\end{align}
for a given choice of $x$ and $z$. 

Recall that any choice $x'\in C_x$ and $ z'\in C_z$ results in the same quantum state in \eqref{eq:csN}. This section describes how we can select coset representatives that will allow us to represent coset states as 
\begin{equation}
\label{eq:csF}
    \ket{W_{x,z}} = \CNOT_J H_I \ket {x+z}.
\end{equation}
The usefulness of this representation will become clear in the following section, which presents an optimal winning strategy. There are two features of the coset representatives $x$ and $z$ we chose that are essential the representation \eqref{eq:csF}: the supports $\supp(x)$ and $\supp(z)$ are disjoint and allow us to derive \eqref{eq:csF} by using the standard (commutation) relations $(I \otimes X)\CNOT = \CNOT(I \otimes X)$,  $(Z \otimes I)\CNOT = \CNOT(Z \otimes I)$, and $ZH=HX$.

\begin{proposition}
    \label{prop:CSGoodReps}
    Given $W \in G(2m, m)$, let $A_W$ be the generator matrix of $W$ in reduced row echelon form. 
    Let $I$ denote the indices of the pivot columns of $A_W$. Then 
    \[
        \CS(W) = \langle e_i \rangle_{i \in I^c}
        ~~\text{and} ~~
        \CS(W^\perp) = \langle e_i \rangle_{i \in I}
    \]
    are coset representatives sets of $W$ and $W^\perp$.
\end{proposition}
\begin{proof}
    Each $e_i$ indexed by $I^c$ is not in $W$, and $W \cap \langle e_i \rangle_{i \in I^c} = \{0\}$ since every non-zero vector in $W$ has a $1$ in some pivot index. Thus, $W + e$ as $e$ varies in $\langle e_i \rangle_{i \in I^c}$ describes all $2^{m}$ cosets of $W$. The same argument holds for $W^\perp$ and $I$.
\end{proof}
\begin{example}
In classical and quantum error correction, the coset representatives of interest are those of minimum weight.
    Consider a matrix of code generators $[5,2]$:\[
	A_W = 
	\begin{bmatrix}
		1 & 0 & 1&1&0\\
		0 & 1 & 1 & 0 & 1
	\end{bmatrix}
\]
Observe that $I=\{1,2\}$ and $I^c=\{3,4,5\}$.
A standard array (in coding theory) for this code is
{\small 
\[
\begin{array}{c|ccc}
\colorbox{red!28}{00000} & 01101	& 10110	& 11011\\
\hline
10000	& 11101	& \colorbox{red!28}{00110}	& 01011\\
01000	& \colorbox{red!28}{00101}	& 11110	& 10011\\
\colorbox{red!28}{00100} & 01001	& 10010	& 11111\\
\colorbox{red!28}{00010} & 01111	& 10100	& 11001\\
\colorbox{red!28}{00001}& 01100	& 10111	& 11010\\
11000	& 10101	& 01110	& \colorbox{red!28}{00011}\\
10001	& 11100	& \colorbox{red!28}{00111}	& 01010
\end{array}
\]}
The first row in the array is the code, and each subsequent row is a different coset. The leftmost column lists the minimum weight coset representatives used in error correction. For the game we consider, we select coset representatives with zeros at the first two positions. Therefore,
Bob's $\CS(W) $ will use the highlighted strings as coset representatives.
\end{example}
\subsection{Coset State Representation}
Using the representation of $\ket{W}$ from \Cref{prop:ketWGoodRep} and the disjoint representatives from \Cref{prop:BCrepForWxz}, we can write
\[\ket{W_{x,z}} = X_{\supp(x)}Z_{\supp(z)}\CNOT_{J}H_{I}\ket{0}^{\otimes 2m}.\]
We now present an equivalent description of the state that benefits Bob and Charlie's strategy.
\begin{proposition}
    \label{prop:BCrepForWxz}
     Given $W \in G(2m, m)$, let $A_W$ be the generator matrix of $W$ in reduced row echelon form. Let $I$ denote the indices of the pivot columns of $A_W$ and let $J=\{(i,j):i \in I, j \in I^c, A_{ij}=1\}$. Given $\langle e_i \rangle_{i \in I}$ and $\langle e_i \rangle_{i \in I^c}$, we have
    \begin{align*}
        \ket{W_{x,z}} &= \CNOT_{J}H_{I}\ket{x+z}
    \end{align*}
\end{proposition}
\begin{proof}
    Consider the following identities:
    \begin{align*}(I \otimes X)\CNOT = \CNOT(I \otimes X)\\
    (Z \otimes I)\CNOT = \CNOT(Z \otimes I)\end{align*}
    In particular, we have $Z_{\supp(z)}\CNOT_J = \CNOT_JZ_{\supp(z)}$ since $\supp(z) \subseteq I^c$ and the controls of $\CNOT_{J}$ are indexed only with $i \in I$. Similarly, $X_{\supp(x)}\CNOT_J  = \CNOT_J X_{\supp(x)}$ since $\supp(x) \subseteq I$ and the target of $\CNOT_J$ are indexed with $i \in I^c$. Thus,
    \[\ket{W_{x, z}} = \CNOT_J X_{\supp(x)} Z_{\supp(z)} H_I \ket{0}^{\otimes 2m}.\] Finally, we have that $X_{\supp(x)}H_I = H_IX_{\supp(x)}$ since $\supp(x) \subseteq I^c$ is disjoint from $I$, and we use the identity $ZH = HX$ to get
    \[\ket{W_{x, z}} = \CNOT_J H_I X_{\supp(x)} X_{\supp(z)} \ket{0}^{\otimes 2m}.\]
    Finally, since we are working over $\F_2$, we have $X_x X_z \ket{0}^{\otimes 2m} = \ket{x + z}$ as desired.
\end{proof}


\section{An Optimal Winning Strategy}\label{sec:optstrat}
This section uses the representations of \Cref{sec:cosetEncoding} to prove that the upper bound in \Cref{thm:pnupperbound} is tight. We first motivate the strategy by considering \Cref{ex:CosetStateIntro}. We then explicitly describe Bob and Charlie's steps to win the game and prove that the winning probability saturates the bound. 

\subsection{Strategy Motivation and Steps
\label{sec:motstrat}}
\begin{example}
\label{ex:N1ExampleContinued}
    To motivate our optimal strategy, we recall \Cref{ex:CosetStateIntro}, which describes the possibilities for $\ket{W_{x, z}}$ when $m = 1$. For now, consider the two cases where Bob and Charlie have no shared entanglement. To reiterate, the best strategy Bob and Charlie can deploy for $W^{(1)}$ will win with certainty by measuring in the computational and Hadamard bases, respectively. Meanwhile, the best strategy Bob and Charlie can deploy for $W^{(2)}$ will win with probability $\frac{1}{2}$. Even though the marginal probabilities of Bob and Charlie guessing the correct $x$ and $z$ separately is $\frac{1}{2}$, they can form a strategy in which their joint probability of winning is also $\frac{1}{2}$. We can formally write down the POVM's for both these cases:
    \begin{align*}
    B^{W^{(1)}}_{00} = \ket{0}\bra{0} & & B^{W^{(2)}}_{00} = \ket{+}\bra{+}\\
    B^{W^{(1)}}_{10} = \ket{1}\bra{1} & & B^{W^{(2)}}_{01} = \ket{-}\bra{-}\\
    C^{W^{(1)}}_{00} = \ket{+}\bra{+} & & C^{W^{(2)}}_{00} = \ket{0}\bra{0}\\
    C^{W^{(1)}}_{01} = \ket{-}\bra{-} & & C^{W^{(2)}}_{10} = \ket{1}\bra{1}
\end{align*}
This table formalizes that, in the case of $W = W^{(1)}$, Bob should guess $00$ if he measures $\ket 0$, and $10$ otherwise. Similarly, Charlie should choose $00$ if he measures $\ket{+}$, and $01$ otherwise. On the other hand, the POVMs for $W = W^{(2)}$ are chosen so that Bob and Charlie either both make an incorrect guess or both make a correct guess.
Recall now the case in which Bob and Charlie are given entangled qubits:
\begin{align*}
    \ket{W_{00,00}^{(3)}} = \frac{1}{\sqrt{2}} \left( \ket{00} + \ket{11} \right) = \ket{\Phi^+}\\
    \ket{W_{00,01}^{(3)}} = \frac{1}{\sqrt{2}} \left( \ket{00} - \ket{11} \right) = \ket{\Phi^-}\\
     \ket{W_{10,00}^{(3)}} = \frac{1}{\sqrt{2}} \left( \ket{01} + \ket{10} \right) = \ket{\Psi^+}\\
     \ket{W_{10,01}^{(3)}} = \frac{1}{\sqrt{2}} \left( \ket{01} - \ket{10} \right) = \ket{\Psi^-}
\end{align*}
\Cref{thm:pnupperbound} gives an upper bound of $\frac{2}{3}$ for Bob and Charlie's winning probability. Assuming this bound is tight, this means that Bob and Charlie should be able to win in the case $W = W^{(3)}$ with probability $\frac{1}{2}$. Indeed, consider the strategy 
\begin{align*}
    B^{W^{(3)}}_{00} = \ket{+i}\bra{+i} & & B^{W^{(3)}}_{10} = \ket{-i}\bra{-i} \\
    C^{W^{(3)}}_{00} = \ket{-i}\bra{-i} & & C^{W^{(3)}}_{01} = \ket{+i}\bra{+i} 
\end{align*}
where $\ket{\pm i} = \frac{\ket{0} \pm i \ket{1}}{\sqrt{2}}$. In this case, one can check that Bob and Charlie both make the correct guess with probability $\frac{1}{2}$.
\end{example}
\Cref{ex:N1ExampleContinued} is surprisingly indicative of the general case. When $m$ is large, the state $\ket{W}$ usually has a great deal of entanglement and there is more opportunity for $Z$-error information to appear on Bob's side, and $X$-information to be on Charlie's side as is the case when $W = W^{(2)}$ in \Cref{ex:N1ExampleContinued}. Even with this added complexity, Bob and Charlie can perform local operations on their own qubits to reduce the state to a simple one that looks like the $m=1$ case. As described in \Cref{sec:upperbound}, it suffices to show for a fixed subspace $W \in G(2m, m)$, that Bob and Charlie can win with a probability of $2^{-k}$ where $k = \dim(\{v_1^m : v \in W \})$ is the dimension of $W$ restricted to the first $m$ coordinates. The displayed algorithm gives an overview of their winning strategy.

\begin{algorithm}[t]
\label{alg:optimalstrat}
\caption{\textbf{Optimal Coset Guessing Game Strategy}}
\begin{algorithmic}
\STATE \textbf{Preparation:}\\ Before starting the game, Bob and Charlie determine a set of circuits they will use for each subspace $W \in G(2m,m)$.
\begin{itemize}[leftmargin=0.5cm]
    \item For each $W \in G(2m,m)$, Bob and Charlie determine the sets $I,J,\CS(W),\CS(W^{\perp})$ as in \Cref{sec:cosetEncoding}.
    \item For each $W \in G(2m,m)$, Bob and Charlie design circuits $\mathcal{B}_W,\mathcal{C}_W$ to be executed on their qubits of \[\ket{W_{x, z}} = \CNOT_J H_I\ket{x + z}\] that
    \begin{enumerate}[leftmargin=0.5cm]
        \item Undo the $\CNOT$ gates that are local to either Bob or Charlie.
        \item Undo $H$ gates on control qubits that have are $\CNOT$-free after the previous step.
        \item Apply $\CNOT$ gates to isolate their entanglement to shared Bell state pairs.
    \end{enumerate}
\end{itemize}
\hrulefill

\STATE \textbf{Game Dynamics:} Bob and Charlie can no longer communicate once the game begins.
\begin{itemize}[leftmargin=0.5cm]
    \item Once Alice sends $\ket{W_{x,z}}$ to Bob and Charlie and announces $W \in G(2m,m)$, Bob and Charlie apply circuits $\mathcal{B}_W,\mathcal{C}_W$ respectively to their qubits.
    \begin{itemize}[leftmargin=0.5cm]
        \item They measure in the $\ket{\pm i} = \frac{1}{\sqrt{2}}(\ket{0} \pm \ket{1})$ basis for qubits in shared Bell state pairs and measure in the computational basis otherwise.
        \item They store their  outputs $b,c$ in classical registers.
    \end{itemize}
    \item Bob and Charlie act on $b,c$ as determined by the circuits $\mathcal{B}_W,\mathcal{C}_W$ to recover guesses $\hat{x},\hat{z}$ for $x,z$ respectively. 
    \begin{itemize}
        \item Bob and Charlie win with probability $2^k$, where $k = \dim(\{v_1^m : v \in W \}) = n-\dim(W \cap \langle e_i \rangle_{i=m+1}^{2m})$.
    \end{itemize}
\end{itemize}
\end{algorithmic}
\end{algorithm}

Recall that the bound is saturated when each player guesses correctly with probability $\Theta(1/2^m)$, and they are simultaneously correct or incorrect. In our strategy, Bob and Charlie process their states locally, so their classical and quantum correlations allow their guesses to be maximally correlated and correct with probability $\Theta(1/2^m)$.

In the coset state $\ket{W_{x,z}} = \CNOT_{J}H_{I}\ket{x+z}$ there will be $\CNOT$ gates that Bob and Charlie can invert locally, depending on $W$, thus removing local entanglement and enabling subsequent inversion of the $H$ gates. The scenario where they can invert all gates locally is equivalent to a classical game in which Alice sends them $x + z$ instead of $\ket{W_{x, z}}$. We address this case in \Cref{subsec:NoEntanglement}. Then, in \Cref{subsec:QuantumStrat}, we show that after locally inverting all gates they can, Bob and Charlie (using only local operations) can reduce the remaining shared entanglement into a set of Bell-state pairs.

\subsection{Bob and Charlie's Entanglement is Locally Removable}
\label{subsec:NoEntanglement}
 We describe how Bob and Charlie can simplify the shared $\ket{W_{x, z}}$ state by performing local operations on their respective qubits. 
 If there is no entanglement between Bob and Charlie after these operations, they can simplify the state to $\ket{x + z}$, reducing the game to a classical one.


\begin{example}
    \label{ex:FirstQuantumExample}
    Suppose $m=2$ and the reduced-row echelon generator matrix for $W$ is
    \begin{equation*}
        A = \begin{pmatrix}
        1 & 1 & 0 & 0 \\
        0 & 0 & 1 & 0 \\
        \end{pmatrix}.
    \end{equation*}
    We have $I = \{1, 3\}$, $\CS(W) = \langle e_2, e_4\rangle$ and $\CS(W^\perp) = \langle e_1, e_3 \rangle$. Suppose Alice chooses \[x = \begin{pmatrix}
        0 & 0 & 0 & 1
        \end{pmatrix} \qquad z = \begin{pmatrix}
        1 & 0 & 1 & 0
        \end{pmatrix}\]
        and so we have \[\ket{W} = \ket{\Phi^+} \otimes \ket{+} \otimes \ket{0} \quad \text{ and } \quad \ket{W_{x, z}} = \ket{\Phi^-} \otimes \ket{-} \otimes \ket{1}\] 
        with $\ket{\Phi^{\pm}} = \frac{1}{\sqrt{2}}\pr{\ket{00} \pm \ket{11}}$. When Alice sends the first $2$ qubits of $\ket{W_{x, z}}$ to Bob and second $2$ qubits to Charlie, they receive states with no shared entanglement. We claim this implies they can reduce their shared state to $\ket{x + z}$. 
        
        First, we show how Charlie can recover $(x + z)_{3}^4$. Charlie knows his first qubit is not entangled with any other qubit, and he knows it must be $\ket{+}$ or $\ket{-}$ when $z_3 = 0$ or $z_3 = 1$, respectively. By applying $H$ to this qubit, he obtains $\ket{0}$ or $\ket{1}$ when $z_3 = 0$ or $1$, turning the qubit into $\ket{z_3}$. Likewise, he knows his second qubit is $\ket{x_4}$ since the fourth column of $A_W$ is all 0's.

        Bob similarly knows that his qubits are entangled and are one of the four bell state pairs depending on $z_1$ and $x_2$. Since the bell state pairs can be encoded using \Cref{prop:ketWGoodRep} and \Cref{prop:BCrepForWxz}, Bob knows his bell state pair is \[\CNOT_{1,2}H_1\ket{z_1x_2}.\] So Bob may remove the entanglement by applying the decoding circuit $H_1\CNOT_{1, 2}$ himself, yielding $\ket{z_1x_2}$.

        As shown above, undoing the operations Alice uses to create $\ket{W}$ is not always possible. For example, a $\CNOT$ gate may be applied to one of Bob's and one of Charlie's qubits. These cases give rise to shared entanglement between Bob and Charlie, which is handled in \Cref{subsec:QuantumStrat}.
\end{example}

Given the state $\ket{W_{x, z}} = \CNOT_{J}H_{I}\ket{x + z}$, the next proposition describes the obvious operations Bob and Charlie can make to undo some of the $\CNOT$ and $H$ operations in $\CNOT_{J}H_{I}\ket{x + z}$ that are local to the first $m$ and second $m$ qubits.

\begin{proposition}
    \label{prop:SeperatingEntanglement}
    Bob and Charlie can win the coset guessing game with probability $p_m$ if they can win it given the state
    \begin{equation}
        \ket{\widehat{W}_{x,z}} = \CNOT_{J'}H_{I'}\ket{x+z}
        \label{eq:en}
         \end{equation}
    where $J'=\{(i,j): 1 \leq i \leq m, m+1 \leq j \leq 2m, (i,j) \in J\}$ and $I' = \{i : (i, j) \in J' \text{ for any } j\}$.
\end{proposition}
\begin{proof}
We show that Bob and Charlie can get from the original game state to the state in \eqref{eq:en} by applying only local gates without communicating.
    Bob and Charlie start with each having $m$ qubits of the state
    \[\ket{W_{x, z}} = \CNOT_{J}H_{I}\ket{x+z}.\] For any pair $(i,j) \in J$ with $1\leq i,j \leq m$, Bob can apply $CNOT_{ij}$. Similarly, for any pair $(i,j) \in J$ with $m+1\leq i, j \leq 2m$, let Charlie apply $\CNOT_{ij}$ to his qubits. Since all $\CNOT$ gates indexed by $J$ commute, the resulting state after applying $\CNOT_{J\setminus J'}$ is
    \begin{align*}
        \CNOT_{J\setminus J'}&\CNOT_{J}H_{I}\ket{x+z} \\
        &= \CNOT_{J\setminus J'}\CNOT_{J \setminus J' } \CNOT_{J'}H_{I}\ket{x+z}\\
        &= \CNOT_{J'}H_{I}\ket{x+z}.
    \end{align*} 
    Finally, for any $i\in I$ such that there are no longer any $\CNOT$ gates with control qubit $i$, the gate $H_i$ can be applied since it will commute with every $\CNOT$ indexed by $J'$. The resulting state after applying $H_{I\setminus I'}$ then has the form
    \begin{align*}
        H_{I \setminus I'}&\CNOT_{J'}H_{I}\ket{x+z} \\
        &= \CNOT_{J'}H_{I \setminus I'}H_{I}\ket{x+z}\\
        &= \CNOT_{J'}H_{I \setminus I'}H_{I \setminus I'} H_{I'}\ket{x+z}\\
        &= \CNOT_{J'}H_{I'}\ket{x+z} = \ket{\widehat{W}_{x,z}}
        \qedhere
    \end{align*}
\end{proof} 
If there is no entanglement between Bob and Charlie, then there is no pair $(i, j) \in J$ with $i \le m < j$. This means the state $\ket{\widehat{W}_{x, z}}$ has become $\ket{x+z}$. Bob and Charlie can reach the desired winning probability by playing the following classical game.


Suppose Alice generates $W \in G(2m, m)$, $x \in \CS(W)$, and $z \in \CS(W^\perp)$ with the representations in \Cref{sec:cosetEncoding} and sends the first $m$ bits of $x + z$ to Bob and the second $m$ bits of $x + z$ to Charlie. We will show that Bob and Charlie can jointly and respectively guess $x$ and $z$ with probability $2^{-k}$ where $k$ is equal to the dimension of $W$ restricted to the first $m$ coordinates. 

We still require that Bob and Charlie do not communicate after receiving their respective bits of $x + z$, so it may seem as if the information about $z \in \CS(W^\perp)$ on Bob's side and the information about $x \in \CS(W)$ on Charlie's side cannot be exchanged. However, our choice of representations for $\CS(W)$ and $\CS(W^\perp)$ as described in \Cref{prop:CSGoodReps} allows Bob and Charlie to perfectly correlate information about the other side's bits without communicating. 

\begin{example}
    \label{ex:BobandCharlieBijection}
    Suppose $m=3$ and the reduced-row echelon generator matrix for $W$ is
\begin{equation*}
    A = \begin{pmatrix}
    1 & 0 & 1 & 0 & 0 & 1\\
    0 & 1 & 1 & 1 & 0 & 1\\
    0 & 0 & 0 & 0 & 1 & 0\\
    \end{pmatrix}.
\end{equation*}
Recall that $\CS(W) = \langle e_3, e_4, e_6\rangle$ and $\CS(W^\perp) = \langle e_1, e_2, e_5 \rangle$. Suppose Alice has 011110, which she splits to send to Bob and Charlie:
\[
\underbrace{011110}_{\text{Alice}} = \underbrace{011}_{\text{Bob}}\underbrace{110}_{\text{Charlie}}
\]
Since Bob and Charlie know and agree on this description of $\CS(W)$ and $\CS(W^\perp)$, they know that \[x + z = z_1z_2x_3x_4z_5x_6.\] 

Bob needs to identify $x$ bits indexed by $I^c$, while Charlie needs to identify $z$ bits indexed by $I$. Since Bob is given the first three bits, he can immediately identify one out of three bits of $x$, and likewise, Charlie can identify one out of three bits of $z$. While Bob and Charlie are left to guess the other two bits, they can improve their joint probability of winning to be higher than choosing uniformly at random. In this case, Bob can guess $x = 001001$ by setting his guess for $x_4$ to the first bit he sees on his side and making his guess for $x_6$ to the second bit on his side. Charlie can similarly guess $z = 100010$ by guessing $z_1$ and $z_2$ to be the bits he sees, but does not need.

In this case, Bob and Charlie are both incorrect in their guesses because they have guessed that $z_1 = x_4$ and $z_2 = x_6$. So, if one person incorrectly guesses a bit, they both do. Neither Bob nor Charlie knows their missing bits, but they can make their guesses coincide. Therefore, their joint probability of winning is equal to the marginal probability of either Bob or Charlie winning. This is the probability that $z_1 = x_4$ and $z_2 = x_6$, which is $\frac{1}{4}$, which agrees with $2^{-k}$ and $k = 2 = |I \cap [m]|$.
\end{example}

Formally, Bob and Charlie agree on a bijection $h_W: I \cap [m] \to I^c \cap [m+1, 2m]$ for every subspace $W$. For every bit $b_i$ Bob sees such that $i \in I^c$, he guesses $x_i = b_i$. For the other possibly nonzero bits of $x$ indexed by $I^c \cap [m+ 1, 2m]$, he guesses $x_{h_W(i)} = b_i$. Charlie will do similarly using $h_W^{-1}$ so that they jointly guess $z_i = x_{h_W(i)}$ for every $i \in I \cap [m]$.

We transfer this classical game into the language of quantum by encoding the classical bits into a state to refer to it later in this section. However, no quantum properties are needed.

\begin{proposition}
    \label{prop:ClassicalStrat} 
    Given $W \in G(2m, m)$, let $A_W$ be the generator matrix of $W$ in reduced row echelon form. Let $I$ denote the indices of the pivot columns of $A_W$ and let $J=\{(i,j):i \in I, j \in I^c, A_{ij}=1\}$. Finally, let $\CS(W) = \langle e_i \rangle_{i \in I^c}$ and $\CS(W^\perp) = \langle e_i \rangle_{i \in I}$.
    
    For every $W \in G(2m, m)$, there exists a bijection $h_W: I \cap [m] \to I^c \cap [m+1, 2m]$. For any such bijection, the POVMs $B_x^W = \bigotimes_{i=1}^mB_i^W$ and $C_z^W = \bigotimes_{j=m + 1}^{2n} C_j^W$ with
    \begin{align*}B_i^W =\begin{cases}
        \ket{x_i}\bra{x_i} & i \in I^c \\
        \ket{x_{h_{W}(i)}}\bra{x_{h_{W}(i)}} & i \in I
    \end{cases} \\
    C_j^W =\begin{cases}
        \ket{z_j}\bra{z_j} & j \in I \\
        \ket{z_{h_{W}^{-1}(j)}}\bra{z_{h_{W}^{-1}(j)}} & j \in I^c
    \end{cases}
    \end{align*}
    satisfy
    \begin{align*}\Tr\bigl[(B_x^W \otimes C_z^W) &\ket{x + z}\bra{x + z}\bigr] \\
    &= \begin{cases}
        1 & z_i = x_{h_W(i)} \text{ for all $i \in I \cap [m]$}\\
        0 & \text{otherwise}
    \end{cases}.\end{align*}
\end{proposition}
The first case in the definition of $B_i$ and $C_i$ outlines what Bob and Charlie should guess, given the bits they receive are on the correct side. The second case correlates their uncertain guesses. For Bob, this can be interpreted as measuring the $i$'th qubit to guess $x_{h_{W}(i)}$. Our claim is that Bob and Charlie win iff $x_i = z_{h_W^{-1}(i)}$.

\begin{proof}
    Our bijection always exists as the size of the domain equals the size of the codomain, both $k$. Splitting up each qubit we may write 
    \begin{align*}
        \mathcal B_i^W &= \bra{x_i + z_{i}}B_i^W\ket{x_i + z_{i}}\\
        \mathcal C_{m + i}^W &= \bra{x_{m + i} + z_{m + i}}C_{m + i}^W\ket{x_{m + i} + z_{m + i}}
    \end{align*}
    and the initial expression becomes
    \begin{align*}
        \Tr\bigl[(B_x^W &\otimes C_z^W) \ket{x + z}\bra{x + z}\bigr] =\\
        &= \Tr\bigl[\bra{x + z}(B_x^W \otimes C_z^W) \ket{x + z}\bigr] \\
        &= \prod_{i = 1}^m \mathcal B_i^W \mathcal C_{m + i}^W
    \end{align*}
    We compute
    \[ \mathcal B_i^W = \bra{x_i + z_{i}}B_i^W\ket{x_i + z_{i}} = \begin{cases}
        \bra{x_i}\ket{x_i}^2 = 1 & i \in I^c\\
        \bra{z_{i}}\ket{x_{h_W(i)}}^2 & i \in I
    \end{cases}\]
    and likewise
     \begin{align*}
     \mathcal C_{m + i}^W &= \bra{x_{m + i} + z_{m + i}}C_{m + i}^W\ket{x_{m + i} + z_{m + i}} \\
     &= \begin{cases}
        \bra{z_{m + i}}\ket{z_{m + i}}^2 = 1 & m + i \in I\\
        \bra{x_{m + i}}\ket{z_{h^{-1}_W(m + i)}}^2 & m + i \in I^c
    \end{cases}.\end{align*}
    If we have $z_i \ne x_{h_W(i)}$, for some $i \in I \cap [m]$, then the product is zero as expected. Otherwise, if $z_i = x_{h_W(i)}$ for all $i \in I \cap [m]$, then we also have $z_{h_W^{-1}(i)} = x_i$ for all $i \in I^c \cap [m + 1, 2m]$ since $h_W$ is a bijection. In this case all terms in the product are $1$ as desired.
\end{proof}

The probability of winning this classical game is then the probability that $z_i = x_{h_W(i)}$ at the indices $I \cap [m]$. Since all the bits are chosen uniformly at random, this occurs with probability $2^{-|I \cap [m]|} = 2^{- k}$ as desired.

\subsection{Bob and Charlie's Entanglement is not Locally Removable}
\label{subsec:QuantumStrat}
This section deals with the case where Bob and Charlie after locally inverting all $\CNOT$ gates they can, still share entanglement. We start with $\ket{\widehat{W}_{x,z}}$ from \Cref{prop:SeperatingEntanglement} and perform another simplification to the state using local operations.

The state \[\ket{\widehat{W}_{x,z}} = \CNOT_{J'}H_{I'}\ket{x+z}\] consists of entangled qubits between Bob and Charlie. Moreover, since the control bits are on Bob's side and the target bits are on Charlie's side, the entanglement in $\ket{\widehat{W}_{x,z}}$ is entirely shared between Bob and Charlie, meaning each of Bob's qubits are either unentangled, or are entangled with qubits on Charlie's side. Unfortunately, we cannot remove this entanglement without applying gates non-locally across Bob's and Charlie's sides. Nonetheless, we find a method to simplify the shared entanglement between Bob and Charlie so that they can achieve an optimal winning strategy.

\begin{example}
    \label{ex:BobandCharlieSharedEntanglement}
    As in \Cref{ex:BobandCharlieBijection}, suppose $m=3$ and the reduced-row echelon generator matrix for $W$ is
\begin{equation*}
    A = \begin{pmatrix}
    1 & 0 & 1 & 0 & 0 & 1\\
    0 & 1 & 1 & 1 & 0 & 1\\
    0 & 0 & 0 & 0 & 1 & 0\\
    \end{pmatrix}.
\end{equation*}
Recall that $\CS(W) = \langle e_3, e_4, e_6\rangle$ and $\CS(W^\perp) = \langle e_1, e_2, e_5 \rangle$. For any choices of $x \in \CS(W)$ and $z \in \CS(W^\perp)$, we have
\begin{align*}
    \ket{W_{x, z}} &= \CNOT_{1, \{3,6\}}\CNOT_{2, \{3,4,6\}}H_{1, 2, 5}\ket{x + z}\\
    \ket{\widehat{W}_{x, z}} &= \CNOT_{1, 6}\CNOT_{2, \{4,6\}}H_{1, 2}\ket{x + z}.
\end{align*}

    where we use the shorthand notation $\CNOT_{i,J} = \prod_{j \in J} \CNOT_{i,j}$ All but the third and fifth qubit are entangled since $\CNOT_{i, j}H_i\ket{x + z}$ entangles the $i$ and $j$'th qubits. We will present operations Bob and Charlie can apply so that each entangled qubit forms an EPR pair with a single qubit on the other side. \rev{In this example, Bob will apply $\CNOT_{2,1}$, and the resulting state will be
    \[\CNOT_{1, 6}\CNOT_{2,4}H_{1,2}\CNOT_{1,2}\ket{x + z}.\]
    We can then interpret $\CNOT_{1, 2}$ as an invertible linear map $\F_2^{2m} \to \F_2^{2m}$ that acts separately on the first $m$ bits of $z$ and the second $m$ bits of $x$, as a block diagonal operator. If $y$ is the image of this map, Bob and Charlie's simplified state is
    \[(\CNOT_{1, 6}H_{1})(\CNOT_{2,4}H_2)\ket{y}\]
    where
    \(y = \begin{pmatrix}
        z_1 & z_1+z_2 & x_3 & x_4 & z_5 & x_6
    \end{pmatrix}\).
    }
\end{example}
Our first goal is to ensure that each qubit on Bob's side is entangled only with qubits on Charlie's side; then, we repeat the same for Charlie. That is, Bob and Charlie will perform local operations so that the pairs in $J'$ satisfy $(i_1, j), (i_2, j) \in J'$ imply $i_1 = i_2$ and $(i, j_1), (i, j_2) \in J'$ imply $j_1 = j_2$.

\begin{proposition}
    \label{prop:SinglingEntanglement}
    Bob and Charlie can win the coset guessing game with probability $p_m$ if they can win it given the state
    \[\CNOT_{i_1, j_1}H_{i_1}\cdots\CNOT_{i_\ell, j_\ell}H_{i_\ell}\ket{y}\]
    where $i_1, \dots, i_\ell \subseteq [m]$ and $j_1, \dots, j_{\ell} \subseteq [m+1..2m]$ are all distinct and $y \in \F_2^{2m}$ satisfies \[y = (f_1(x_1^m + z_1^m), f_2(x_{m + 1}^{2m} + z_{m + 1}^{2m}))\] under linear maps $f_1, f_2: \F_2^{2m} \to \F_2^{m}$ known to Bob and Charlie satisfying $f_1(x_1^m) = x_1^m$ and $f_2(z_{m + 1}^{2m}) = z_{m + 1}^{2m}$.
\end{proposition}
\begin{proof}
     We will show how Bob and Charlie can apply gates to their qubits so that entanglement is limited to pairs of qubits $(i,j)$, where $1\leq i \leq m, m+1\leq j \leq 2m$. Consider the following identities:
\begin{align*}
    \CNOT_{jk} \CNOT_{ij} \CNOT_{ik} &= \CNOT_{ij} \CNOT_{jk}\\
    \CNOT_{ij} \CNOT_{ik} \CNOT_{jk} &= \CNOT_{jk} \CNOT_{ij}\\
    (H_i \otimes H_j) \CNOT_{ij} &= \CNOT_{ji} (H_i \otimes H_j).
\end{align*}
We want to understand how the gates $\CNOT_{i,i'}$ for $i,i' \in I$ and $\CNOT_{i,i'}$ for $i,i' \in I^c$ pass through $\CNOT_{J}$. Using the identities above, for every distinct $i,i' \in I$, we know $\CNOT_{i,i'}$ commutes with every $\CNOT$ gate in $\CNOT_J$ except for those of the form $\CNOT_{i',j}$. Let $J_{i'} = \{(i',j) \in J\}$ and $J_{i'}(i) = \{(i,j):(i',j) \in J_{i'}\}$, then
\begin{align*}
    \CNOT_{i,i'}&\CNOT_J \\
    &= \CNOT_{i,i'}\CNOT_{J_{i'}}\CNOT_{J \setminus J_{i'}}\\
    &= \CNOT_{J_{i'}}\CNOT_{J_{i'}(i)}\CNOT_{i,i'}\CNOT_{J \setminus J_{i'}}\\
    &= \CNOT_{J_{i'}}\CNOT_{J_{i'}(i)}\CNOT_{J \setminus J_{i'}}\CNOT_{i,i'}\\
    &= \CNOT_{J(A')}\CNOT_{i,i'}
\end{align*}
For a matrix $A_W$ in reduced row echelon form, we denote the set of pivot columns as $I(A)$ and $J(A) = \{(i,j):A_{ij}=1,i<j\}$. Likewise, given sets $I = \{i_1,i_2,\ldots, i_l\}\subseteq [2m]$ and $J$ consisting of pairs of indices $(i,j)$ with $i \in I$ and $i<j$, we denote by $A(I,J)$ the $m$-by-$2m$ matrix with $A_{k,i_k}=1$, $A_{ij}=1$ for $(i,j) \in J$ and zero elsewhere. In the above equation, $J(A')$ is the set obtained from the matrix $A'$, \rev{where given $i=i_k$ and $i'=i_{k'},$ $A'$ is obtained by adding the $k'$-th row of $A_W$ to the $k$-th row of $A_W$ modulo $2$ in the non-pivot columns of $A_W$.} Similarly, suppose $i,i' \in I^c$. Then $\CNOT_{i,i'}$ commutes with every $\CNOT$ gate in $\CNOT_J$ except for those of the form $\CNOT_{j,i}$. Let $J_{i} = \{(j,i) \in J\}$ and $J_{i}(i') = \{(j,i'):(j,i) \in J_{i'}\}$, then
\begin{align*}
    \CNOT_{i,i'}&\CNOT_J
    \CNOT_{i,i'}\CNOT_{J_{i}}\CNOT_{J \setminus J_{i}}\\
    &= \CNOT_{J_{i}}\CNOT_{J_{i}(i')}\CNOT_{i,i'}\CNOT_{J \setminus J_{i}}\\
    &= \CNOT_{J_{i}}\CNOT_{J_{i}(i')}\CNOT_{J \setminus J_{i}}\CNOT_{i,i'}\\
    &= \CNOT_{J(A')}\CNOT_{i,i'}
\end{align*}
where $J(A')$ is the set obtained from the  $A'$, and $A'$ is the matrix obtained by adding the $i$-th column of $A_W$ to the $i'$-th column of $A_W$ modulo $2$. This allows us to perform Gaussian elimination on the non-pivot columns of $A_W$ up to swapping columns. In particular, we can ensure at most a single $1$ exists in each non-pivot column, each in a distinct row. 

By \Cref{prop:SeperatingEntanglement}, we can assume $I \subseteq [m]$ and $J$ consists of pairs $(i,j)$ with $1\leq i \leq m$, $m+1 \leq j \leq 2m$, i.e. in the matrix $A_W$ denoting control-target qubit pairs, all pivots lie in the first $m$ columns, and all non-zero non-pivot columns lie in the last $m$ columns.

Bob and Charlie thus proceed as follows: they agree on a Gaussian elimination procedure to reduce $A_W$ to a matrix $A'$ where each non-zero non-pivot column in the last $n$ indices has at most a single $1$ entry. Bob applies the appropriate gates $\CNOT_{J_1}$ to his qubits, and Charlie applies the appropriate gates $\CNOT_{J_2}$ to his qubits:
\begin{align*}
    \CNOT_{J_2}&\CNOT_{J_1}\CNOT_JH_I \ket{x+z} \\
    &= \CNOT_{J(A')}\CNOT_{J_2}\CNOT_{J_1}H_I \ket{x+z}\\
    &= \CNOT_{J(A')}H_I\CNOT_{J_2}\CNOT_{J'_1} \ket{x+z}
\end{align*}
where $J'_1 = \{(j,I):(i,j) \in J_1\}$, and $J(A')$ has the form $\{(i_1,j_1),(i_2,j_2),\ldots, (i_\ell,j_\ell)\}$ stated in the hypothesis. For any index $i \in I$ but $i \not \in \{i_1,\ldots, i_\ell\}$, Bob can apply the Hadamard gate $H_i$ so that the processed state is now $\CNOT_{J(A')}H_{I'}\CNOT_{J_2}\CNOT_{J'_1} \ket{x+z}$, where $I'=\{i_1,\ldots,i_\ell\}$.

Since $J'_1$ and $J_2$ are local to Bob and Charlie respectively and $\supp(x) \cap \supp(z) = \emptyset$, we may write
\begin{align*}&\CNOT_{J_2}\CNOT_{J'_1}\ket{x + z} \\
&\quad  = \CNOT_{J'_1} \ket{x_1^m + z_1^m} \otimes \CNOT_{J_2}\ket{x_{m + 1}^{2m} + z_{m + 1}^{2m}} \\
&\quad = \ket{y}\end{align*}
where $y \in \F_2^{2m}$. Since $\CNOT_{1, 2}(\ket{b_1,b_2}) = \ket{b_1,b_1 + b_2}$, we can interpret both $\CNOT_{J'_1}$ and $\CNOT_{J_2}$ as invertible linear maps on $x + z \in \F_2^{2m}$. Moreover, since $J'_1$ only has indices on Bob's side and $J_2$ only has indices on Charlie's side, they can be interpreted as diagonal block linear maps on the first and second half of coordinates of $\F_2^{2m}$. Thus, we have $f_1, f_2: \F_2^{2m} \to \F_2^{m}$ that effect bits indexed by $J'_1, J_2$. By our constructions, we have $J'_1 \subseteq (I\cap [m])^2$ and similarly $J_2 \subseteq (I^c \cap [m+1..2m])^2$. This means the $x$ bits on Bob's side do not change and the $z$ bits on Charlie's side do not change:\[f_1(x_1^m) = x_1^m \text{ and } f_2(z_{m + 1}^{2m}) = z_{m + 1}^{2m}.\] 
Since $J'_1$ and $J_2$ can be calculated from $A_W$, the bijections $f_1, f_2$ sending $x + z \mapsto y$ are known to Bob and Charlie.
\end{proof}

In general, the resulting state 
\[\CNOT_{i_1, j_1}H_{i_1}\cdot \ldots\cdot\CNOT_{i_\ell, j_\ell} H_{i_\ell} \ket{y}\]
is a collection of computational qubits and EPR pairs. There are $\ell$ bell state pairs, and the other computational qubits occur at all indices outside of $\{i_1, \dots, i_\ell,  j_1, \dots, j_\ell\}$. The specific string of bits and bell state pairs depends on $y$. 

Now Bob and Charlie are ready to measure their qubits. For all qubits with indices outside $\{i_1, \dots, i_\ell,  j_1, \dots, j_\ell\}$, the measurements Bob and Charlie should make are exactly the same as in \Cref{subsec:NoEntanglement}, except now Bob and Charlie need $y$ instead of $x + z$. After obtaining $y$, they can use the bijection $y \mapsto x + z$.

For the bell pairs in the state indexed by $(i, j) \in \{(i_1, j_1), \dots, (i_t, j_t)\}$, Bob and Charlie should measure each pair in the $\ket{\pm i}$ basis as in \Cref{ex:N1GameIntro}. \rev{Since neither Bob nor Charlie can be certain about the expected measurement outcomes they obtain from a bell pair, Bob and Charlie will not always be able to correlate their guesses classically. However, the choice of measurement on the bell pair ensures that their measurement outcomes are correlated.} Bob and Charlie can exploit this to correlate their guesses on the outcome of measuring each bell state pair. Measuring in the $\ket{\pm i}$ basis reduces Bob and Charlie's task to guessing which one of two bell pairs was measured. That is, their guesses will depend on the probability their measurement corresponds to their guess, which can always be made to be $\frac{1}{2}$.
\begin{theorem}
    \label{thm:UpperBoundIsTight}
    Bob and Charlie can win the coset guessing game with a probability 
    \[p_m = \frac{1}{\binom{2m}{m}_2}\sum_{k=0}^m 2^{k^2} \binom{m}{k}_2^2 \left(\frac{1}{2}\right)^k = \Theta \left ( \frac{1}{2^m} \right ).\] 
    The strategy they use to achieve the bound is given below.
    \begin{enumerate}[leftmargin=0.5cm]
    \item Bob and Charlie reduce their shared state to that of \Cref{prop:SinglingEntanglement}, yielding the linear maps $f_1, f_2,$ and $f := (f_1, f_2)$.
    \item For each $x \in \CS(W)$ Bob computes $y_B = f(x)$ and for each $z \in \CS(W^\perp)$ Charlie computes $y_C = f(z)$.
    \item For each $W \in G(2m, m)$, they agree on a bijection $h_W:$ 
    \[(I \cap [m]) \setminus \{i_1, \dots, i_\ell\} \to (I^c \cap [m + 1, 2m]) \setminus \{j_s\}_{s=1}^\ell.\]
    \item They measure using the POVMs
    \[B_{x}^W = \bigotimes_{i = 1}^m B_i^W \quad \text{ and } \quad C_z^W = \bigotimes_{j = m + 1}^{2m} C_j^W\]
    where 
    \begin{align*}B_i^W = 
\begin{cases}
\hspace{-0.4cm}
\begin{split}
    &\ket{(y_B)_{i}} \bra{(y_B)_{i}} \\
    & \hspace{1.5cm} i  \not \in \{i_s\}_{s=1}^\ell \text{ and } i \in I^c\\
    &\ket{(y_B)_{h_W(i)}} \bra{(y_B)_{h_W(i)}}\\
    & \hspace{1.5cm} i  \not \in \{i_s\}_{s=1}^\ell \text{ and } i  \in I\\
    &\ket {+i} \bra{+i} \\
    & \hspace{1.5cm} i_t \in \{i_s\}_{s=1}^\ell \text{ and } (y_B)_{j_t} = 0\\
    &\ket {-i} \bra{-i} \\
    & \hspace{1.5cm}  i_t \in \{i_s\}_{s=1}^\ell \text{ and } (y_B)_{j_t} = 1
\end{split}
\end{cases}
\end{align*} and 
\begin{align*} C_j^W =
\begin{cases}
\hspace{-0.4cm}
\begin{split}
    &\ket{(y_C)_j} \bra{(y_C)_j} \\
    &  \hspace{1.5cm} j \not \in \{j_s\}_{s=1}^\ell \text{ and } j \in I\\
    &\ket{(y_C)_{h_{W}^{-1}(j)}} \bra{(y_C)_{h_{W}^{-1}(j)}} \\
    &  \hspace{1.5cm} j \not \in \{j_s\}_{s=1}^\ell \text{ and } j \in I^c\\
    &\ket {+i} \bra{+i} \\
    &  \hspace{1.5cm} j \in \{j_s\}_{s=1}^\ell \text{ and } (\widehat{y}_C)_{h_{W}^{-1}(j)} = 1\\
    &\ket {-i} \bra{-i} \\
    &  \hspace{1.5cm} j \in \{j_s\}_{s=1}^\ell \text{ and } (\widehat{y}_C)_{h_{W}^{-1}(j)} = 0.
\end{split}
\end{cases}
\end{align*}
    \end{enumerate}
\end{theorem}
\begin{proof}
    Following the proof of \Cref{thm:pnupperbound}, it suffices to show that $p_m(W) = 2^{-k}$ where $k$ is the dimension of $W$ intersected with the last $m$ coordinates. Assume Bob and Charlie share the state
    \[\CNOT_{i_1, j_1}H_{i_1}\cdots\CNOT_{i_\ell, j_\ell}H_{i_\ell}\ket{y}\]
    and know the linear maps $f_1,f_2$ such that \[f_1(y_1^m) = x_1^m + f_1(z_1^m) \; \text{ and } \; f_2(y_{m + 1}^{2m}) = f_2(x_{m + 1}^{2m}) + z_{m + 1}^{2m}.\]
    For each $1 \le i \le m < j \le 2m$ we compute
    \begin{align*}
        \label{eq:yiequation}
         y_i &=
    \begin{cases}
        (x_1^m)_i + f_1(z_1^m)_i =& f_1(z_1^m)_i = (y_C)_i \\
        &\quad i \in I \cap [m]\\
        (x_1^m)_i + f_1(z_1^m)_i =& x_i = f_1(x_1^m)_i = (y_B)_i \\
        &\quad i \in I^c \cap [m]
    \end{cases}
    \,\\
    y_j &=
    \begin{cases} 
        f_2(x_{m + 1}^{2m})_j + (z_{m + 1}^{2m})_j =& z_j = f(z_{m + 1}^{2m})_j = (y_C)_j \\
        &\quad j \in I \cap [m + 1, 2m]\\
        f_2(x_{m + 1}^{2m})_j + (z_{m + 1}^{2m})_j =& f_2(x_{m + 1}^{2m})_j = (y_B)_j \\
        &\quad j \in I^c \cap [m + 1, 2m].
    \end{cases}
    \end{align*}
    
    We split their shared state into qubits that are entangled and those that are not. Denoting $\ket{\phi_{i,j}}=\CNOT_{i,j}H_i\ket{y_{i,j}}$ and $\ket{\phi} =  \CNOT_{i_1, j_1}H_{i_1}\cdots\CNOT_{i_\ell, j_\ell}H_{i_\ell}\ket{y}$, we have
    \begin{align*}
    \Tr&\sb{\pr{B_x^W \otimes C_z^W} \ket{\phi} \bra{\phi}}
         = \bra{\phi}(B_x^W \otimes C_z^W)\ket{\phi}\\
        &= \prod_{(i, j) \in \{(i_s, j_s)\}_{s=1}^\ell} \bra{\phi_{i,j}}\pr{B_i^W \otimes C_j^W}\ket{\phi_{i,j}}\cdot\\
        &\quad \prod_{\substack{i \in [n] \setminus \{i_s\}_{s=1}^\ell\\ j \in [n + 1, 2n] \setminus \{j_s\}_{s=1}^\ell}} \bra{y_i}B_i^W\ket{y_i} \bra{y_j}C_j^W\ket{y_j}.
    \end{align*}
    
    First, we deal with qubits without entanglement. By our choice of $y_B$ and $y_C$, we have
    \begin{align*}
         \bra{y_i}B_i^W\ket{y_i} &=
    \begin{cases}
    \hspace{-0.4cm}
    \begin{split}
        \bra{(y_C)_i} &\ket{(y_B)_{h_W(i)}}^2 \\
        &i \in (I \cap [m]) \setminus \{i_s\}_{s=1}^\ell\\
        \bra{(y_B)_i} &\ket{(y_B)_{i}}^2 = 1 \\
        &i \in I^c \cap [m]
    \end{split}
    \end{cases}
    \,\\
    \bra{y_j}C_j^W\ket{y_j} &= 
    \begin{cases} 
    \hspace{-0.4cm}
    \begin{split}
        \bra{(y_C)_j} &\ket{(y_C)_j}^2=1  \\
        &j \in I \cap [m + 1, 2m]\\
        \bra{(y_B)_j}&\ket{(y_C)_{h_{W}^{-1}(j)}}^2 \\
        &j \in (I^c \cap [m + 1, 2m]) \setminus \{j_s\}_{s=1}^\ell.
     \end{split}
    \end{cases}
    \end{align*}
    The product of these two values is $1$ iff $(y_B)_i = (y_C)_{h_W(i)}$ since $h_W$ is chosen to restrict to a bijection on the indices
    \[(I \cap [m]) \setminus \{i_s\}_{s=1}^\ell \to (I^c \cap [m+1, 2m]) \setminus \{j_s\}_{s=1}^\ell.\]
    Furthermore, since $f$ is linear map, we have
    \[y = f(x + z) = f(x) + f(z) = y_B + y_C.\]
    Thus, Bob and Charlie win iff $y_i = y_{h_W(i)}$. This is exactly the condition in \Cref{prop:ClassicalStrat}, except we are replacing $x + z$ with $y$. Even so, we know $f: x + z \mapsto y$ is a bijection, and thus, the bits of $y$ are uniformly distributed when $x$ and $z$ are chosen uniformly at random. The probability that $y_i = y_{h_W(i)}$ holds on all indices $(I \cap [m]) \setminus \{i_s\}_{s=1}^\ell$ is thus $2^{-k + \ell}$.

    For the entangled qubit case, for all $i \in \{i_s\}_{s=1}^\ell \subseteq I \cap [m]$ and $j \in \{j_s\}_{s=1}^\ell \subseteq I^c \cap [m + 1, 2m]$ we have
    \[y_i = (y_C)_i \quad \text{ and } y_j = (y_B)_j.\]
    Now suppose $i = i_t$ and $j = j_t$ for some $t \in [\ell]$. By our choice of $B_i$ and $C_j$, we have
    \begin{align*}
        B_i = 
        \begin{cases}
            \ket{+i}\bra{+i} & y_{j} = 0\\
            \ket{-i}\bra{-i} & y_{j} = 1\\
        \end{cases}       
  \quad
        C_j = 
        \begin{cases}
            \ket{+i}\bra{+i} & y_{i} = 1\\
            \ket{-i}\bra{-i} & y_{i} = 0
        \end{cases}
    \end{align*}
    Recalling that $CNOT_{i, j}H\ket{y_{i, j}}$ can take on any of the four EPR pairs, we compute 
    \begin{align*}
    \bra{y_{i, j}} &H_{i} \CNOT_{i, j}\pr{B_i^W \otimes C_j^W}\CNOT_{i, j} H_{i} \ket{y_{i, j}}\\
    &= \begin{cases}
        \bra{\Phi^+} \pr{\ket{+i}\bra{+i} \otimes \ket{-i}\bra{-i}}\ket{\Phi^+} & \ket{y_{i, j}} = \ket{00}\\
        \bra{\Phi^-} \pr{\ket{+i}\bra{+i} \otimes \ket{+i}\bra{+i}}\ket{\Phi^-} & \ket{y_{i, j}} = \ket{10}\\
        \bra{\Psi^+} \pr{\ket{-i}\bra{-i} \otimes \ket{-i}\bra{-i}}\ket{\Psi^+} & \ket{y_{i, j}} = \ket{01}\\
        \bra{\Psi^-} \pr{\ket{-i}\bra{-i} \otimes \ket{+i}\bra{+i}}\ket{\Psi^-} & \ket{y_{i, j}} = \ket{11}\\
    \end{cases}\\
    &= \frac{1}{2}
    \end{align*}
    in every case. This means the probability that Bob and Charlie guess both bits correctly is $\frac{1}{2}$ for all $\ell$ EPR pairs. Combining this with the probability they guess all other bits correctly, we get $2^{-k + \ell} \cdot 2^{-\ell} = 2^{-k}$ as desired.
\end{proof}

Bob and Charlie's strategy is entirely different for the qubits that end up being Bell-state pairs. \rev{While it still holds that Bob guesses $x$ iff Charlie guesses $z$, they will win only if their measurement outcome correctly describes the Bell-state pairs.}

Suppose after Bob and Charlie make their guesses, Alice tells Bob and Charlie if they have guessed correctly. If there are no Bell state pairs, then Bob and Charlie know that their outputs are deterministic, so they can trust Alice's response. For example, if Alice plays the game for multiple rounds with the same $W$, $x$, and $z$, Bob and Charlie learn nothing new after the first round. However, if Alice tells Bob and Charlie they are incorrect $N$ consecutive rounds, they only know their guess was incorrect with probability $N2^{-\ell}$. This means Bob and Charlie gain less information from Alice's response as the number of entangled qubits increases.
\section{Conclusion}
We considered a stricter version of a coset guessing game introduced and analyzed in \cite{coladangelo2021hidden, QK:VidickZ21, CV22} and derived an upper bound below the previously known one. We then presented an optimal strategy to achieve this bound. While developing an optimal strategy, we devised an encoding circuit using only $\CNOT$ and Hadamard gates, \rev{which builds CSS codes from EPR pairs using only local operations.} We found that the role of quantum information that Alice communicates to Bob and Charlie is to make their responses correlated rather than improve their individual (marginal) correct guessing rates. 


It remains open to determine asymptotically optimal strategies for the general game. We did not address the question of \textit{rigidity} (the uniqueness of optimal strategies) or \textit{robustness} (the closeness of nearly optimal strategies to the optimal). Both properties are relevant to the certification of devices in device-independent quantum cryptography (DIQC). Results about implementing monogamy-of-entanglement games into quantum key distribution (QKD) protocols have been studied in \cite{cervero2023device}. However, in contrast to the coset monogamy game, our coset guessing game cannot be considered a monogamy-of-entanglement game since the presence of Charlie does not lower Bob's probability of winning.

Subsequent work from this paper has derived stronger asymptotic upper bounds than previously known on the winning rate of a more general coset monogamy game setup \cite{schleppy2025winning}.

\section*{Acknowledgments}
We thank G.~L.~Matthews, A.~Paul, J.~Shapiro, and N.~Verma for the discussions at the 2024 VT-Swiss Coding Theory \& Cryptography Summer School \& Collaboration Workshop.

\section*{Appendices}
\subsection{Some Standard Quantum Gates \label{sec:gates}}
	Unitary operators acting on one or more qubits are often called {\it quantum gates}.  The Pauli and Hadamard gates are the most commonly used single-qubit gates, and the $\CNOT$ is the most commonly used two-qubit gate.
\\[1ex]
\ul{The Identity and the Pauli Matrices/Gates}:\\[0.8ex]
The matrices are given below, together with their action on the basis vectors in the quantum circuit notation:
\begin{small}
		\begin{center}
                    \resizebox{3in}{1in}{
				\begin{tikzpicture}
			\node at (0,0) {$I = \begin{bmatrix}
			 1 & 0\\	0 &  1 \end{bmatrix}$};
		 	\node[right] at (0.75,0.35) {\Qcircuit @C=1em @R=.7em {
		 		&&& \lstick{\ket{0}}     & \gate{I} & \rstick{\ket{0}}\qw}};
		 \node[right] at (0.75,-0.35) {\Qcircuit @C=1em @R=.7em {
		 		&&& \lstick{\ket{1}}     & \gate{I} & \rstick{\ket{1}}\qw}};
				\begin{scope}[shift={(0,-1.85)}]
			\node at (-0.15,0) {$ X = \begin{bmatrix}
				0 &  1 \\ 1 & 0 \end{bmatrix}$};
			\node[right] at (0.75,0.35) {\Qcircuit @C=1em @R=.7em {
					&&& \lstick{\ket{0}}     & \gate{X} & \rstick{\ket{1}}\qw}};
			\node[right] at (0.75,-0.35) {\Qcircuit @C=1em @R=.7em {
					&&& \lstick{\ket{1}}     & \gate{X} & \rstick{\ket{0}}\qw}};
				\end{scope}
				\begin{scope}[shift={(5.5,0)}]
			\node at (-0.15,0) {$ Y = \begin{bmatrix}
				0 & -i \\ i & 0 \end{bmatrix}$};
		\node[right] at (0.75,0.35) {\Qcircuit @C=1em @R=.7em {
					&&& \lstick{\ket{0}}     & \gate{Y} & \rstick{i\ket{1}}\qw}};
		\node[right] at (0.75,-0.35) {\Qcircuit @C=1em @R=.7em {
					&&& \lstick{\ket{1}}     & \gate{Y} & \rstick{-i\ket{0}}\qw}};
			\end{scope}
			\begin{scope}[shift={(5.5,-1.85)}]
			\node at (-0.15,0) {$ Z = \begin{bmatrix}
				1 &  0 \\ 0 & -1 \end{bmatrix}$};
		\node[right] at (0.75,0.35) {\Qcircuit @C=1em @R=.7em {
					&&& \lstick{\ket{0}}     & \gate{Z} & \rstick{\ket{0}}\qw}};
		\node[right] at (0.75,-0.35) {\Qcircuit @C=1em @R=.7em {
					&&& \lstick{\ket{1}}     & \gate{Z} & \rstick{-\ket{1}}\qw}};
				\end{scope}
			\end{tikzpicture}
            }
							\vspace{-2ex}
		\end{center}
\end{small}
Note that $ X$ maps $\ket{0}$ into $\ket{1}$ and vice versa, and thus it is often referred to as the quantum {\tt NOT} or the bit-flip gate. 
\\[1ex]
\ul{The Hadamard gate}:\\[1ex]
The Hadamard gate is defined by the normalized Hadamard $2\times 2$ matrix. The matrix is given below, together with its action on the basis vectors in the quantum circuit notation:
\\
\begin{small}
\begin{tikzpicture}
				\node at (-3,0) {$H = \frac{1}{\sqrt{2}}\begin{bmatrix}
					1 &  1 \\ 1 & -1 \end{bmatrix}$};
				\node at (0,0.35) {\Qcircuit @C=1em @R=.7em {
						&&& \lstick{\ket{0}}     & \gate{H} & \rstick{(\ket{0}+\ket{1})/\sqrt{2}=\ket{+}}\qw}};
				\node at (0,-0.35) {\Qcircuit @C=1em @R=.7em {
						&&& \lstick{\ket{1}}     & \gate{H} & \rstick{(\ket{0}-\ket{1})/\sqrt{2}=\ket{-}}\qw}};
				\end{tikzpicture}
\end{small}
\\
The basis of $\mathbb{C}^2$ defined by vectors 
$\ket{+}$ and $\ket{-}$ that results from the Hadamard action of the computational basis is known as the Hadamard basis.
\\[1.5ex]
\ul{The \texttt{CNOT} Gate}:\\[1ex]
The two-qubit quantum gate known as  quantum \texttt{XOR} or controlled-not gate \texttt{CNOT} is specified as a map of the bases vectors and as follows:

\begin{flushleft}
        \resizebox{3in}{0.5in}{
	\begin{tikzpicture}
	\node[right] at (0,0.35) {
		$\texttt{CNOT}: |x,y\rangle \rightarrow |x,x\oplus y\rangle$};
	\node[right] at (1.1,-0.35) {
		$x,y\in\{0,1\}$};
	\node[right] at (5,0) {\Qcircuit @C=1em @R=1em {
			\lstick{\ket{x}} & \ctrl{1} & \rstick{\ket{x}} \qw \\
			\lstick{\ket{y}} & \targ & \rstick{\ket{x\oplus y}} \qw
	}};
	\end{tikzpicture}
        }
\end{flushleft}
We refer to $x$ as the control bit and $y$ as the target. The target bit is flipped if the control bit equals $1$. Similarly, we say that qubit $\ket{x}$ is the control qubit, and qubit $\ket{y}$ is the target qubit. Observe that the target qubit undergoes the $X$ (NOT) action if the control qubit is $\ket{1}$.
\\[1.5ex]
We can create Bell states by $H$ and $\CNOT$ gates as follows:
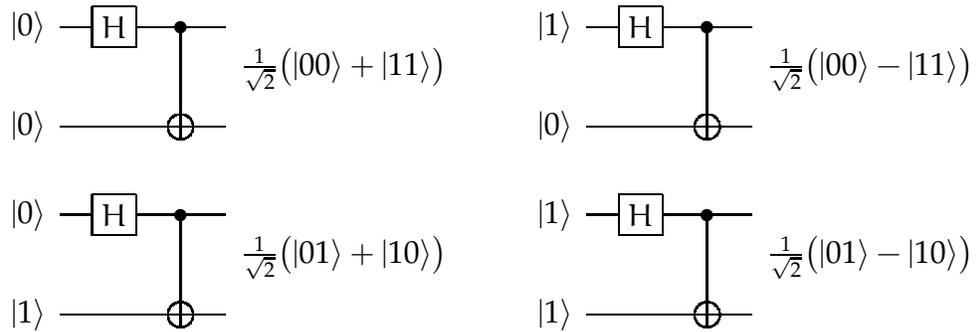
\begin{center} 
\resizebox{2.45in}{1in}{
\begin{tikzpicture}
\node at (0,0) {\Qcircuit @C=1em @R=.7em {
  \ket{0} & & \gate{H} & \ctrl{3} & \qw &\\
  & & & & \rstick{\frac{1}{\sqrt{2}}\bigl(\ket{00}+\ket{11}\bigr)}\\
  & \\
  \ket{0} & & \qw & \targ & \qw &
}};
\node at (0,-2.5) {\Qcircuit @C=1em @R=.7em {
  \ket{0} & & \gate{H} & \ctrl{3} & \qw &\\
  & & & & \rstick{\frac{1}{\sqrt{2}}\bigl(\ket{01}+\ket{10}\bigr)}\\
  & \\
  \ket{1} & & \qw & \targ & \qw &
}};
\begin{scope}[xshift=1.25cm]
\node at (5.75,0) {\Qcircuit @C=1em @R=.7em {
  \ket{1} & & \gate{H} & \ctrl{3} & \qw &\\
  & & & & \rstick{\frac{1}{\sqrt{2}}\bigl(\ket{00}-\ket{11}\bigr)}\\
  & \\
  \ket{0} & & \qw & \targ & \qw &
}};
\node at (5.75,-2.5) {\Qcircuit @C=1em @R=.7em {
  \ket{1} & & \gate{H} & \ctrl{3} & \qw &\\
  & & & & \rstick{\frac{1}{\sqrt{2}}\bigl(\ket{01}-\ket{10}\bigr)}\\
  & \\
  \ket{1} & & \qw & \targ & \qw &
}};
\end{scope}
\end{tikzpicture}
}
\end{center}
\subsection{Positive Operator-Valued Measure (POVM)}
Quantum measurement extract classical information from quantum states. We define a general measurement as follows:
\begin{itemize}
	\item A set of  positive-semidefinite
 operators $\{E_i\}$ such that $\sum_iE_i=I$.
	\item  For input $\ket{\psi}$, output $i$ happens wp $\bra{\psi}E_i\ket{\psi}$, and $\ket{\psi}$ collapses to  $\frac{1}{\sqrt{\bra{\psi}E_i\ket{\psi}}}E_i\ket{\psi}$.
\end{itemize}
We most often consider positive-semidefinite $E_i$ that are projections on non-orthogonal vectors. 


\bibliographystyle{IEEEtran}
\bibliography{citation}

\end{document}